\documentclass[11pt]{article}
\usepackage{amssymb,amsmath,amsthm,amsfonts}
\usepackage{graphicx}
\usepackage{subfigure}
\usepackage[usenames, dvipsnames]{color}
\usepackage{mathrsfs}
\usepackage[CJKbookmarks, colorlinks, bookmarksnumbered=true,pdfstartview=FitH,linkcolor=black]{hyperref}
\usepackage{tcolorbox}
\def\disp{\displaystyle}

\def\crr{\cr\noalign{\vskip2mm}}

\def\dref#1{(\ref{#1})}

\theoremstyle{plain}
\newtheorem{theorem}{Theorem}[section]
\newtheorem{lemma}{Lemma}[section]
\newtheorem{proposition}{Proposition}[section]
\newtheorem{corollary}{Corollary}[section]

\numberwithin{equation}{section}

\theoremstyle{definition}
\newtheorem{definition}{Definition}

\newtheorem{remark}{Remark}[section]

\newcommand{\R}{{\mathbb R}}

\def\A{\mathcal{A}}
\def\B{\mathcal{B}}

\begin{document}

\title{{\bf
Delay    Compensation  for  Regular   Linear Systems    \footnote{\small
This work is supported by  the National Natural Science
Foundation of China, No.  61873153. }} }

\author{ Hongyinping Feng
 \footnote{\small
 Email: fhyp@sxu.edu.cn.}
   \crr {\it  School of Mathematical Sciences
  Shanxi University,  Taiyuan, Shanxi, 030006,
China}
 }

 \maketitle

\begin{abstract}
This is the third part of four series papers,
 aiming at the    delay    compensation  for  the  abstract  linear system $(A,B,C)$.
 Both the input   delay and  output  delay are investigated.
   We first propose a full state feedback control    to  stabilize
 the  system   $(A,B)$
 with   input   delay and then   design   a Luenberger-like observer  for  the system $(A,C)$
 in terms of the delayed output.
 We  formulate the   delay  compensation  in  the framework of       regular linear systems.
The developed approach builds upon   an  upper-block-triangle transform that is associated with a Sylvester operator equation.
 It is found   that the controllability/observability  map of      system $(-A,B)$/$(-A,-C)$
  happens to be the solution of the corresponding  Sylvester equation. As an immediate  consequence,
    both the feedback law and the state  observer can be expressed  explicitly in the operator form.
The exponential  stability of the resulting closed-loop system and the exponential  convergence of the  observation error are
established without using the Lyapunov functional approach.
 The  theoretical results are validated through  the   delay  compensation  for a benchmark one-dimensional wave equation.

\vspace{0.3cm}


\vspace{0.3cm}

\end{abstract}
\section{Introduction}
It is well  known that
the time-delay is ubiquitous in   engineering practices.
Since  the Smith predictor was introduced in \cite{1959Smith},  a fair amount of research results
 about the  delay  compensation have been done  for finite-dimensional systems.
However,   the control  of   infinite-dimensional systems with  time-delay
  is still a challenging problem and
 the corresponding
  results
 are much less than that for finite-dimensional ones.
In \cite{Krstic2009SCLdelay}, \cite{SCLJieQi2019delay} and  \cite{JMWang2017},
the input  delay  is  compensated for the  reaction-diffusion equation by the
 method of partial differential equation (PDE) backstepping.  These results
 can be considered as   more or less the extensions of       delay compensation  for the
 ordinary differential equations (ODEs)  discussed in \cite{Krstic2008delay} and  \cite{KrsticAnnual-Control}.
 When   there are only finite
 unstable modes in the  open-loop system, the input  delay can be compensated by the
 finite-dimensional spectral truncation  technique.
 See for instance   \cite{PrieurandTrelat2019TAC} and \cite{Prieur2020}.


Although arbitrarily small
  delay in the feedback may  destroy the stability of the  system \cite{Datkodekay},
some delays are still  helpful to the system stability.
 When the output delay happens to be the   propagation time, the wave equation can be stabilized by a
 delayed non-collocated  boundary displacement feedback \cite{FengTACdisplacement}.
When the output delay  equals   even multiples of the  propagation time,
 a direct feedback can stabilize the wave equation exponentially \cite{WangSIAMdelay}.
Even the wave equation with nonlinear boundary condition can be stabilized by  the  positive effect of the  delay \cite{FengSIAM2}.

  Stabilizations for one-dimensional   wave and beam equations with arbitrarily long
output delays are discussed in \cite{XuCZ2012} and \cite{YangKYauto} where
the problem is  solved by  both   observer and predictor:
The  state is estimated in the time span where the observation is available;
and the state is predicted in the time interval where the observation is not available.
 Very recently, the   idea used in \cite{XuCZ2012} and \cite{YangKYauto} has been  extended to an abstract   linear systems in \cite{MeiZDTAC,Meizd}  and \cite{MeiZDSCL}.
 However,  the systems considered in \cite{MeiZDTAC} are only limited to
  the conservative  system   and even the common  unstable finite-dimensional linear systems
  do not belong  to such class.
Although the systems  studied in  \cite{MeiZDSCL} can be unstable,
the bounded    control operator must be required.

Since the   delay dynamics  are usually dominated by a transport equation \cite{XuGQ2006delay}, the problem of input or output delay  compensation for infinite-dimensional   systems can be  described by a PDE-PDE
cascade system. In contrast with the  ODE-PDE or PDE-ODE cascade, the
  control of   PDE-PDE cascade is  much  more complicated and the corresponding results are still   fairly
   scarce. Some results about this topic  can be found in the  monograph 	
\cite{KrsticDelaybook}.

  In this paper, we consider the      delay     compensation     for
 general  abstract linear systems.  Both  the input  delay  and output    delay are considered systematically.
  Let  $(A,B,C)$  be a linear system with the state space $Z$, input space $U$ and the output space
$Y$.
%
 The problem is described by
\begin{equation} \label{wxh20202131701}
 \left.\begin{array}{l}
\disp  \dot{z} (t) = Az(t)+B   u(t-\tau ),  \ \
\disp y(t)=Cz(t-\mu),
\end{array}\right.
\end{equation}
where   $y(t)$ is the measured output,  $u(t)$ is the control input, both of them  are
  delayed by $\mu$ and $\tau$ units of time, respectively.
We will study  the input and output    delay  compensation
separately. There are two key issues. The first one is about the stabilization
of system \dref{wxh20202131701}  by  the  state feedback, and the second one
is on the design of   state observer for
system \dref{wxh20202131701} in terms of the delayed output $y(t)$.
Thanks  to    the separation principle of
the linear systems, the output feedback law  of system \dref{wxh20202131701} is almost trivial once
 these two   key issues are addressed.
The developed  approach is systematic  and can  be applied to the general   regular linear systems which cover the common transport  equations,  reaction-diffusion equations, wave equations  and   the Euler-Bernoulli beam equations.


By writing the  delay dynamics as a transport equation,
the delay compensation  for system \dref{wxh20202131701}
then amount  to controlling  or observing a PDE-PDE cascade.
In this paper, the main idea of the  PDE-PDE cascade treatment comes   from  the well known fact that the  upper-block-triangle matrix
 can be decoupled  as  a  block-diagonal  matrix  by an  upper-block-triangle transformation
 that is associated with a
Sylvester matrix  equation. More precisely,
\begin{equation} \label{20207131306}
\begin{pmatrix}
I&S\\
0&I
\end{pmatrix}
\begin{pmatrix}
A_1&Q \\
0&A_2
\end{pmatrix}\begin{pmatrix}
I & S\\
0&I
\end{pmatrix}^{-1}
=
\begin{pmatrix}
A_1&0 \\
0&A_2
\end{pmatrix},
\end{equation}
where  $A_1$, $A_2$ and $Q$ are matrices with  appropriate dimensions,
   $I$ is the  identity matrix
 on  appropriate dimensional   spaces
and
 $S$ is the solution of  the Sylvester equation $A_1S-SA_2=Q$.
 Owing to the block-diagonal structure, either stabilization or observer design of  the transformed system  ${\rm diag}(A_1,A_2)$ is much  simpler than the original upper-block-triangle matrix.
  We will  treat the   delay  compensation  for general regular linear systems by following
 this idea.
 In our previous studies  \cite{FPart1} and \cite{FPart2},
  this idea
    has been used  to compensate the actuator dynamics and sensor dynamics  for abstract linear systems.
  Different from \cite{FPart1} and \cite{FPart2} where at least one of $A_1$ and $A_2$ is required to be bounded,
   the   delay  compensation for infinite-dimensional systems considered in this paper
  always  leads to a   PDE-PDE cascade system.
 Generally speaking, the Sylvester operator equation
with unbounded operators  is hard to be solved.
Fortunately, we find that
the  controllability/observability  map of system $(-A,B)$/$(-A,-C)$    happens to be the
solution of  corresponding    Sylvester operator equation.
As a result, the upper-block-triangle transformation that decouples the  PDE-PDE cascade  system    can be obtained explicitly.


The paper is organized as follows. In Section \ref{Preliminary}, we present  some preliminaries on
the regular linear systems.  Section \ref{shiftsemigroup} investigates the
vanishing shift semigroup which is used to  describe the   delay dynamics.
Section  \ref{Sylvesterinput} considers  a Sylvester operator equation that is crucial  to the input  delay compensation. The state feedback is proposed to
stabilize   system  $(A,B)$ with the input  delay in Section \ref{Actuator}.
 Sections \ref{OupputSyl} and \ref{observer} are devoted to the sensor delay  compensation.
The    Sylvester operator equation  that is used to output
  delay compensation is considered in Section \ref{OupputSyl} and the Luenberger-like   observer is designed in terms of the delayed output in Section \ref{observer} where the exponentially   convergence
of the observer is also proved.  In Section \ref{Application}, the developed approaches are  applied to a one-dimensional wave equation to validate the theoretical results.
 For easy  readability, some results that are less relevant to the delay   compensator design are arranged  in the Appendix.



\section{Background on regular linear systems} \label{Preliminary}
This section presents  a brief overview of the regular linear system theory.
We only summarize  the results that will be used in the  sections thereafter.
We refer  the  interested reader to   the references
 \cite{TucsnakWeiss2009book,Weiss1989regular,Weiss1994MCSS,Weiss1994TAMS} and \cite{Weiss1997TAC} for more details.
  We  first  introduce the definition of   dual   space    with
respect to a pivot space   that has been discussed extensively in \cite{TucsnakWeiss2009book}  and is
crucial  in the theory of
unbounded control and observation.

Suppose that $X$ is a Hilbert space and  $A : D(A )\subset X \to X $ is a densely defined operator  with $\rho(A) \neq \emptyset$.
The operator $A$ can determine two Hilbert spaces:
$(D(A), \|\cdot\|_1)$ and $([D(A^*)]', \|\cdot\|_{-1})$, where
$ [D(A^*)]' $ is the dual space of $ D(A^*) $ with respect to  the  pivot space $X$, and the norms $\|\cdot\|_1$ and $\|\cdot\|_{-1}$ are defined   by
\begin{equation} \label{20191141722}
 \left\{\begin{array}{l}
 \disp \|x\|_1=\|(\beta-A)x\|_X,\ \ \forall\ x\in D(A),\crr
 \disp \|x\|_{-1}=\disp \|(\beta-A)^{-1}x\|_X,   \ \ \forall\ x\in X,
 \end{array}\right. \ \ \beta\in\rho(A).
\end{equation}
These two spaces are independent of the choice of $\beta\in\rho(A)$ since   different  choices
of $\beta$ lead  to  equivalent norms.
 For brevity, we denote
the two spaces   as
$ D(A) $ and $ [D(A^*)]' $ in the sequel.
 The adjoint of $A^*\in \mathcal{L}(D(A^*),X)$, denoted by  $\tilde{A}$, is defined   as
 \begin{equation} \label{20191121602}
 \disp     \langle \tilde{A} x,y\rangle_{ [D(A^*)]', D(A^*)}=
 \langle x,A ^*y\rangle_{X },\ \ \forall\ x\in X,\  y\in D(A ^*).
\end{equation}
It is evident that $\tilde{A} x=Ax$ for any $x\in D(A)$. So
  $\tilde{A}\in \mathcal{L}(X,  [D(A^*)]')$ is an extension  of $A $. Since $A$ is
densely defined,  such an extension is unique. By \cite[Proposition 2.10.3]{TucsnakWeiss2009book},
 we have  $(\beta-\tilde{A})\in \mathcal{L}(X,[D(A^*)]')$  and
$(\beta-\tilde{A})^{-1}\in \mathcal{L}( [D(A^*)]',X)$ which imply that
$\beta-\tilde{A}$ is an isomorphism from $X$ to $[D(A^*)]' $.

Suppose that $Y$ is the output Hilbert space and  $C\in \mathcal{L}(D(A),Y)$. The
$\Lambda$-extension of $C $  with respect to $A $  is defined by
   \begin{equation} \label{20205281601}
 \left. \begin{array}{l}
 \disp C_{ \Lambda}x=\lim\limits_{\lambda\rightarrow +\infty}
C \lambda(\lambda  - {A} )^{-1}x,\ \ \forall\ x\in  D (C_{ \Lambda})=\{x \in X \ |\  \mbox{the   limit exists}\}.
\end{array}\right.
\end{equation}
Define the norm
\begin{equation} \label{20205281609}
  \|x\|_{D(C_{\Lambda})}= \|x\|_X+\sup_{\lambda\geq\lambda_0}
  \|C \lambda(\lambda  - {A} )^{-1}x\|_Y,\ \ \forall\ x\in D(C_{\Lambda}),
\end{equation}
where $\lambda_0\in \R$ such that $[\lambda_0, \infty ) \subset \rho(A) $.
Then, it follows from  \cite[Proposition 5.3]{Weiss1994MCSS} that  $D(C_{ \Lambda})$ with norm   $\|\cdot\|_{D(C_{\Lambda})}$ is
   a Banach space and
$C_{ \Lambda}\in \mathcal{L}(  D(C_{ \Lambda}),Y)$. Moreover,
  we have the continuous embeddings:
 \begin{equation} \label{20205281604}
 D(A) \hookrightarrow  D(C_{ \Lambda})\hookrightarrow  X   \hookrightarrow [ D (A ^*)]'.
\end{equation}

\vskip 0.5cm
The following results  are brought  from \cite{Weiss1997TAC}:
\begin{proposition}\label{Regular}
Let $X$, $U$ and $Y$ be the state space, input space and the output space, respectively.
The triple
 $(A,B,C)$ is said to be  a regular linear system    if and only if
the following assertions hold  true:

  (i)~ $A$ generates a $C_0$-semigroup $e^{At}$ on $X$;

 (ii)~ $B\in \mathcal{L}(U,[D(A^*)]')$ and  $C\in \mathcal{L}(D(A), Y)$ are admissible for the $C_0$-semigroup $e^{At}$;

 (iii)~ $C_{\Lambda}(s -\tilde{A})^{-1}B$ exists for
some (hence, for every ) $s\in \rho(A)$;

(iv)~
 $s\rightarrow \|C_{\Lambda}(s -\tilde{A})^{-1}B\|$ is bounded on some right half-plane.
  \end{proposition}
 \begin{definition} \label{De20204202110}
  Let $X$ and $U$  be Hilbert spaces,    let $A$  be
the generator of a $C_0$-semigroup $e^{At}$ on $X$
and  let $B\in \mathcal{L}(U,[D(A^*)]')$. Then,
$F\in \mathcal{L}(D(A), U)$ stabilizes system  $(A,B)$ exponentially
   if  the following assertions hold:

(i)~  $(A, B, F)$ is a regular triple;

(ii)~  there exists  an $s\in\rho(A)$ such that $I$   is an admissible feedback operator for    $F_{\Lambda}(s-\tilde{A})^{-1}B$;

(iii)~  $A+B F_{\Lambda}$ generates an exponentially stable $C_0$-semigroup $e^{(A+B F_{\Lambda})t}$ on $X$.
\end{definition}
\begin{definition}\label{De20204191058}
  Let $X$ and $Y$  be the Hilbert spaces, and let   $A$  be
the generator of a $C_0$-semigroup $e^{At}$ on $X$ and
  $C\in \mathcal{L}(D(A),Y)$. Then,  $L\in \mathcal{L}(Y,[D(A^*)]')$
 detects system $(A,C)$  exponentially
 if  the following assertions hold  true:

(i)~  $(A,L,C)$ is a regular triple;

(ii)~  there exists an $s\in\rho(A)$ such that $I$   is an admissible feedback operator for    $C_{\Lambda}(s-\tilde{A})^{-1}L$;

(iii)~  $A+L C_{\Lambda}$ generates an exponentially stable $C_0$-semigroup $e^{(A+L C_{\Lambda})t}$ on $X$.
\end{definition}

\section{Vanishing shift semigroup} \label{shiftsemigroup}
It is well known that  the  time-delay  dynamics can be modeled   as  a transport equation which is  usually  associated with
 a vanishing shift semigroup. In this section,  we introduce  some preliminaries on
shift semigroup that is  useful for       delay compensations.

Let $U$ be a Hilbert space with norm $\|\cdot\|_U$.  For any $\alpha>0$,
we denote by  $ L^2([0,\alpha];U)$
the Hilbert space of
the  measurable and square integrable functions from $[0,\alpha]$ to $ U$.
The inner product is
 \begin{equation} \label{20205302222}
  \langle \phi_1,\phi_2\rangle_{ L^2([0,\alpha];U)  }
 =\int_0^{\alpha}\langle  \phi_1(x) ,\phi_2(x)\rangle_{U}dx,\ \ \forall\ \phi_1,\phi_2\in L^2([0,\alpha];U)  .
\end{equation}
 Define the operator $ G_{\alpha}  :D( G_{\alpha}  )\subset  L^2([0,\alpha];U)  \to  L^2([0,\alpha];U)  $ by
\begin{equation} \label{20205292012}
 \left.\begin{array}{l}
 \disp (G_{\alpha}   f)(x)=- \frac{d}{dx}f(x) ,\ \ \forall\ f\in D( G _{\alpha} )=\left\{f\in H^1([0,\alpha];U)\ |\ f (0)=0\right\}.
\end{array}\right.
\end{equation}
Then,  $G_{\alpha}$ generates a vanishing  right shift  semigroup $e^{G_{\alpha}t}$ on
$ L^2([0,\alpha];U)  $,   given by
\begin{equation} \label{202062716}
 \left(e^{G_{\alpha}t} f \right)(x)=\left\{\begin{array}{ll}
  f(x-t),&  x-t\geq0,\\
  0,&x-t<0,
\end{array}\right.\ \ \forall\ f\in L^2([0,\alpha];U).
\end{equation}
The adjoint of $G_{\alpha} $ is
\begin{equation} \label{20205302138}
 \left.\begin{array}{l}
\disp  (G_{\alpha}^*   f)(x)=  \frac{d}{dx}f(x) ,\ \ \forall\ f\in D( G _{\alpha}^* )=\left\{f\in H^1([0,\alpha];U)\ |\ f (\alpha)=0\right\},
\end{array}\right.
\end{equation}
which generates a vanishing  left  shift  semigroup
\begin{equation} \label{202062754}
 \left(e^{G_{\alpha}^*t} f \right)(x)=\left\{\begin{array}{ll}
  f(x+t),&  x+t\leq\alpha,\\
  0,&x+t>\alpha,
\end{array}\right.\ \ \forall\ f\in L^2([0,\alpha];U).
\end{equation}
Obviously, both $e^{G_{\alpha}t}$ and $e^{G_{\alpha}^*t}$   are exponentially stable on $L^2([0,\alpha];U)$.

Let $[H^1([0,\alpha];U)]'$ be the dual  space of $ H^1([0,\alpha];U) $ with respect to the pivot space
$L^2([0,\alpha];U)$.
Define  the operator $B_{\alpha}: U \to [H^1([0,\alpha];U)]'$ by
 \begin{equation} \label{20205302223}
  \langle B_{\alpha} u,f \rangle_{[H^1([0,\alpha];U)]', H^1([0,\alpha];U)}=
  \langle  u,f(0)\rangle_{U },\ \ \forall\ f \in H^1([0,\alpha];U), \ \forall\ u\in U
\end{equation}
and the operator $C_{\alpha}:   D(C_{\alpha})\subset  L^2([0,\alpha];U)   \to U  $ by
  \begin{equation} \label{20205302118}
 C_{\alpha}f  =f (\alpha),\ \ \forall\ f \in D(C_{\alpha})=H^1([0,\alpha];U).
\end{equation}
\begin{lemma} \label{lm202062747}
 For any $\alpha>0$, let $G_{\alpha}$, $B_{\alpha}$ and $C_{\alpha}$ be defined by \dref{20205292012}, \dref{20205302223}
 and \dref{20205302118}, respectively.
 Then, both
  $B_{\alpha}\in \mathcal{L}(U ,[D(G_{\alpha} ^*)]')$ and
  $C_{\alpha}\in \mathcal{L}( D(G_{\alpha}  ),U)$   are admissible for the vanishing right shift semigroup $e^{G_{\alpha}t}$ and
   \begin{equation} \label{202062750}
 (\lambda-\tilde{G}_{\alpha} )^{-1}B_{\alpha}=E_{\lambda},\ \
 \lambda\in  \mathbb{C},
\end{equation}
where
the operator  $E_{\lambda}\in \mathcal{L}( U,  D(C_{\alpha}) )$  is given  by
 \begin{equation} \label{2020732059}
 E_{\lambda}u=e^{- \lambda  x}u,\ x\in [0,\alpha],\ \  \forall\ u\in U.
\end{equation}
Moreover,
$(G_{\alpha}, B_{\alpha},C_{\alpha})$ is a regular linear system.
\end{lemma}
\begin{proof}
It follows from \dref{20205302223} that the adjoint of $B_{\alpha}$
 satisfies
$B_{\alpha}^*\in \mathcal{L}( D(G_{\alpha} ^*),U)$ and
$B_{\alpha}^*f  =f (0)$ for any $ f \in D(G_{\alpha}^*) $. From   \dref{202062754},
we deduce that
 \begin{equation} \label{202074952}
\int^{\alpha}_0\|B_{\alpha}^*e^{G_{\alpha}^*t}f\|_U^2dt=
\int^{\alpha}_0\| f(t)\|_U^2dt=\|f\|^2_{L^2([0,\alpha];U)},\ \ f\in D(G_{\alpha}^*),
\end{equation}
which implies that  $B_{\alpha}^*$ is admissible for $e^{G_{\alpha}^*t}$ and thus,
$B_{\alpha} $ is admissible for $e^{G_{\alpha} t}$. Similarly, it follows from   \dref{202062716} and \dref{20205302118} that
$C_{\alpha}\in \mathcal{L}( D(G_{\alpha}  ),U)$ is admissible for $e^{G_{\alpha}t}$.

By a straightforward computation,   it follows that $\rho(G_{\alpha})=\rho(G_{\alpha}^*)=\mathbb{C}$ and
\begin{equation} \label{202074959}
\begin{array}{l}
\left\langle  (\lambda-\tilde{G}_{\alpha} ) E_{\lambda}u,\phi\right\rangle_{[D(G_{\alpha}^*)]',D(G_{\alpha}^*)}=
\left\langle E_{\lambda}u,\left(\overline{\lambda}- {G}_{\alpha}^*\right) \phi\right\rangle_{L^2([0,\alpha];U)}\crr
\disp =\int_0^{\alpha}\langle e^{- \lambda  x}u, \overline{\lambda}\phi(x)\rangle_Udx-
\int_0^{\alpha}\left\langle e^{- \lambda  x}u,  \frac{d}{dx}\phi(x)\right\rangle_Udx\crr
\disp ={\lambda}\int_0^{\alpha}\langle e^{- \lambda  x}u,  \phi(x)\rangle_Udx
+\langle u,\phi(0)\rangle_U - \lambda
\int_0^{\alpha}\langle e^{- \lambda  x}u, \phi(x)\rangle_Udx\crr
=\langle u,\phi(0)\rangle_U,\ \ \forall\ u\in U, \ \phi\in  D(G_{\alpha}^*), \ \lambda\in \mathbb{C},
\end{array}
\end{equation}
which, together with \dref{20205302223}, leads to  $ (\lambda-\tilde{G}_{\alpha} ) E_{\lambda}=B_{\alpha}$.
This means that  \dref{202062750} holds.
By  \dref{20205302118}, \dref{202062750} and \dref{2020732059},  we  conclude that
$C_{\alpha} (\lambda-\tilde{G}_{\alpha} )^{-1}B_{\alpha}u=C_{\alpha}E_{\lambda}u
=e^{-\lambda\alpha}u$ for any $u\in U$. This  implies that
$ C_{\alpha} (\lambda-\tilde{G}_{\alpha} )^{-1}B_{\alpha}\in \mathcal{L}( U)$
and $\lambda\to \|C_{\alpha} (\lambda-\tilde{G}_{\alpha} )^{-1}B_{\alpha}\|$
 is bounded on some right
half-plane.   Hence, $(G_{\alpha}, B_{\alpha},C_{\alpha})$ is a regular linear system.
\end{proof}

As in  \cite[Section 2.2]{Salamon1987},
we define a subspace of  $L^2([0,\alpha];U)$ by
\begin{equation} \label{202062291621}
G_{B_{\alpha}}=\left\{ f\in  L^2([0,\alpha];U) \ |\  \tilde{G}_{ \alpha}f+B_{\alpha}u \in L^2([0,\alpha];U), u\in U  \right\}.
\end{equation}
For any $f\in G_{B_{\alpha}}$, there exists a $u_f \in U$ such that  $ \tilde{G}_{ \alpha }f+B_{\alpha}u_f\in L^2([0,\alpha];U)  $.
This shows that  $\tilde{G}_{\alpha}^{-1}(\tilde{G}_{ \alpha }f+B_{\alpha}u_f) =
  f+\tilde{G}_{\alpha}^{-1}B_{\alpha}u_f
  \in D(G_{\alpha})$ and hence $f\in D(G_{\alpha})+\tilde{G}_{\alpha}^{-1}B_{\alpha}U$.
 So we obtain $G_{B_{\alpha}}\subset D(G_{\alpha})+\tilde{G}_{\alpha}^{-1}B_{\alpha}U$.
For any $g=g_1+\tilde{G}_{\alpha}^{-1}B_{\alpha}u_g \in D(G_{\alpha})+\tilde{G}_{\alpha}^{-1}B_{\alpha}U $ with $g_1\in D(G_{\alpha})$ and $u_g\in U$,
a simple computation shows that
$ \tilde{G}_{ \alpha}g+B_{\alpha}(-u_g)=  \tilde{G}_{ \alpha}g_1 \in
 L^2([0,\alpha];U)$,
which means that   $g\in  G_{B_{\alpha}}  $  and hence $D(G_{\alpha})+\tilde{G}_{\alpha}^{-1}B_{\alpha}U \subset G_{B_{\alpha}} $. Therefore,
 \begin{equation} \label{202062291827}
  G_{B_{\alpha}}= D(G_{\alpha})+\tilde{G}_{\alpha}^{-1}B_{\alpha}U.
 \end{equation}
By
  \cite[Section 2.2]{Salamon1987}, $G_{B_{\alpha}}$  with inner product
  \begin{equation} \label{20206291645}
\| f\|_{ G_{B_{\alpha}}}^2= \| f\|_{L^2([0,\alpha];U)}^2+
\|u_f\|_U^2+\| \tilde{G}_{ \alpha }f+B_{\alpha}u_f \|^2_{L^2([0,\alpha];U)}
\end{equation}
  is   a Hilbert space, where $u_f\in U$ such that
  $ \tilde{G}_{\alpha} f+  B_{\alpha}u_f
  \in L^2([0,\alpha];U) $.   It follows from Lemma \ref{lm202062747}
that $-\tilde{G}_{\alpha}^{-1}B_{\alpha}U=\{ c_u\ |\ c_u(x)\equiv u,\ x\in[0,\alpha], u\in U \}$ which, together with
   \dref{20205292012}, \dref{202062291827}  and \dref{20205302118}, gives
   \begin{equation} \label{20206292003}
  G_{B_{\alpha}}= H^1([0,\alpha];U)=D(C_{\alpha}).
 \end{equation}

\section{Sylvester equation associated with input  delay}
\label{Sylvesterinput}
In this section, we consider a Sylvester equation that is    closely related to  the input
 delay compensation.
Let    $A$ be  a  generator of the $C_0$-semigroup
 $e^{A t}$  on $Z$. Suppose that
    $B\in \mathcal{L}(U,[D(A^*)]')$ is admissible for $e^{-At}$.  Then,
 $B^*\in \mathcal{L}(D(A^*),U)$ is admissible for $e^{-A^*t}$.
By exploiting  \cite[Proposition 4.3.4, p.124]{TucsnakWeiss2009book},
  $B^*e^{A^* (\cdot-\tau  )} h=B^*e^{-A^* ( \tau-\cdot  )} h\in H^1([0,\tau  ];U)$
for any  $h\in D(A^*)$ and $\tau>0$.  As a consequence, we can  define
the operator $S_{\tau  } :  [H^1([0,\tau  ];U)]'   \to [D(A^*)]'$  by
\begin{equation} \label{20205312113}
 \langle S_{\tau  } f, z\rangle_{ [D(A^*)]',D(A^*)} =
\left \langle f,  B^*e^{A^* (\cdot-\tau  )} z \right\rangle_{[H^1([0,\tau  ];U)]' , H^1([0,\tau  ];U)  }
\end{equation}
for any  $z\in D(A^*)$ and $f\in [H^1([0,\tau  ];U)]' $.
Suppose that $g\in  L^2([0,\tau ];U)$. Then
\begin{equation} \label{20205312127}
  \begin{array}{l}
\disp\langle S_{\tau   } g, z\rangle_{ [D(A^*)]',D(A^*)}=
\left\langle g,  B^*e^{A^* (\cdot-\tau  )} z \right\rangle_{ L^2([0,\tau ];U) }  =
   \int_0^{\tau   }\left\langle  g(x),B^*e^{A^* (x-\tau  )}z \right\rangle_{U}dx\crr
\disp \hspace{2cm}=\left\langle  \int_0^{\tau   }e^{\tilde{A} (x-\tau   )}B g(x)dx,z \right\rangle_{[D(A^*)]',D(A^*)}
,\ \  \forall\  z\in D(A^*)\subset Z  .
\end{array}
  \end{equation}
Since   $B$ is  admissible  for $e^{-At}$, we have $\int_0^{\tau   }e^{\tilde{A} (x-\tau   )}B g(x)dx\in Z$  which, together with \dref{20205312127}, implies that
   $S_{\tau  }\in \mathcal{L}(  L^2([0,\tau ];U), Z)$ and
  \begin{equation} \label{20205302131}
 S_{\tau  } g=\int_0^{\tau  }e^{\tilde{A} (x-\tau  )}B g(x)dx,\ \ \ \forall\ g\in  L^2([0,\tau ];U) .
 \end{equation}
This means that $S_{\tau  } $ is an extension of  the controllability map of system $(-A,B)$,
as defined in
\cite[Definition 4.1.3, p.143]{Curtain1995}.
\begin{lemma}\label{lm20205302114}
Let  $(A,B,K)$ be  a regular linear system with the state space $Z$, input space $U$ and the output  space $U$.
 Suppose that $G_{\tau  }$,  $B_{\tau  }$ and $C_{\tau  }$ are  defined by \dref{20205292012},
\dref{20205302223}  and  \dref{20205302118} with $\alpha=\tau$, respectively.
 Then,   the operator $S_{\tau  }\in \mathcal{L}(  L^2([0,\tau ];U) , Z)$  defined by \dref{20205312113} satisfies:
  \begin{equation} \label{20205302207}
 S_{\tau  }B_{\tau  }= e^{ -\tilde{A}\tau   }B \in \mathcal{L}(U,[D(A^*)]')
 \end{equation}
and
 \begin{equation} \label{20206271839}
 \left\{\begin{array}{l}
\disp \tilde{A}S_{\tau  }f-S_{\tau  }\tilde{G}_{\tau  }f=BC_{\tau }f,\crr
\disp  K_{\Lambda}e^{A\tau } S_{\tau  }f \in U,
\end{array}\right.
\ \ \forall\ \ f\in G_{B_{\tau}} ,
\end{equation}
 where  $G_{B_{\tau}}$  is defined by \dref{202062291621} with $\alpha=\tau$.
 \end{lemma}
\begin{proof}
  By  \dref{20205312113} and \dref{20205302223}, we deduce
\begin{equation} \label{20205311511}
 \begin{array}{rl}
\disp
 \left \langle S_{\tau  }B_{\tau  } u,       h
  \right\rangle_{[D(A^*)]',D(A^*)}&\disp =
 \left \langle B_{\tau  } u,  B^*e^{A^* (\cdot-\tau  )} h \right\rangle_{[H^1([0,\tau  ];U)]' , H^1([0,\tau  ];U)  }=
    \left\langle  u,B^*e^{-A^*  \tau   }h\right\rangle_{U }  \crr
    &\disp = \left\langle  e^{-\tilde{A}   \tau   }Bu,h\right \rangle_{[D(A^*)]',D(A^*) },\ \
   \forall\ \  u\in U,\ \ h\in D(A^*),
    \end{array}
\end{equation}
which leads to   \dref{20205302207} easily. It follows from \dref{20205292012},  \dref{20205302118} and \dref{20205302131}  that
 \begin{equation} \label{2020531919}
 \begin{array}{rl}
 S_{\tau  }G_{\tau  } g  &\disp =
 - \int_0^{\tau  } e^{\tilde{A} (x-\tau  )}B g  '(x)dx=
 -   B g  (\tau  )+ \tilde{A}\int_0^{\tau  } e^{\tilde{A} (x-\tau  )}B g  (x)dx\crr
 &\disp =-BC_{\tau  }g +\tilde{A}S_{\tau  }g ,\ \ \ \ \forall\ g\in  D(G_{\tau  } ),
 \end{array}
 \end{equation}
 which implies that the Sylvester equation $   \tilde{A}S_{\tau  }-  S_{\tau  }G_{\tau  }  =BC_{\tau  } $ holds on $D(G_{\tau})$.
For any $f\in G_{B_{\tau}} $, by
  \dref{202062291827},
  $f$ can be divided into two parts $f=g_f +\tilde{G}_{\alpha}^{-1}B_{\alpha}u_f$, where
  $g_f\in D(G_{\tau})$ and $u_f\in U$. By Lemma \ref{lm202062747},  $\tilde{G}_{\tau}^{-1}B_{\tau}u_f=-E_0u_f\equiv -u_f$.
Owing to \dref{202062291827}, \dref{2020531919} and $g_f\in D(G_{\tau})$,   the first equation of \dref{20206271839} holds  if we can prove that
\begin{equation} \label{20206292221}
 \disp  \tilde{A}S_{\tau} E_0 u_f- S_{\tau}\tilde{G}_{\tau} E_0 u_f
 =B C_{\tau}E_0 u_f  .
 \end{equation}
Actually, it  follows from \dref{20205302207} that
\begin{equation} \label{20207181018}
 -S_{\tau}\tilde{G}_{\tau} E_0u_f= S_{\tau}\tilde{G}_{\tau} (\tilde{G}_{\tau}^{-1}B_{\tau}u_f )=S_{\tau  }B_{\tau  }u_f
=e^{-\tilde{A}\tau}Bu_f.
\end{equation}
 By \dref{20205302118} and \dref{20205302131}, it follows that
 \begin{equation} \label{20206292214}
  B C_{\tau}E_0 u_f-\tilde{A}S_{\tau} E_0 u_f= Bu_f-\tilde{A}\int_0^\tau
 e^{\tilde{A} (x-\tau  )}B u_fdx=e^{-\tilde{A}\tau}Bu_f,
 \end{equation}
which, together with  \dref{20207181018},  leads to  \dref{20206292221}  easily.
Therefore, the first equation of \dref{20206271839} holds.

Now, we prove the remaining part  of \dref{20206271839}.
Since
$(A,B,K)$ is  regular, we have
 \begin{equation} \label{2020819929}
 K_{\Lambda}\int_0^{\cdot }e^{\tilde{A} (\cdot-s)}Bg( s)  ds\in H^1_{\rm loc}([0,+\infty);U),\ \ \forall\ g\in H^1_{\rm loc}([0,\infty);U).
  \end{equation}
  In particular,
  \begin{equation} \label{2020819937}
 K_{\Lambda}\int_0^{\tau }e^{\tilde{A} (\tau-s)}Bg( s)  ds\in  U ,\ \ \forall\ g\in H^1([0,\tau];U).
  \end{equation}
 For any $f\in G_{B_{\tau}}$,
it follows from \dref{20206292003} that
 $f(\tau-\cdot)\in   H^1([0,\tau];U) $. Since $B$ is admissible for $e^{-At}$,
 \begin{equation} \label{2020819920}
e^{A\tau} S_{\tau  } f= e^{A\tau}\int_0^{\tau  }e^{\tilde{A} (x-\tau  )}Bf(x)  dx =
  \int_0^{\tau  }e^{\tilde{A} x}Bf(x)  dx=\int_0^{\tau  }e^{\tilde{A} (\tau-x)}Bf(\tau-x)  dx
,
 \end{equation}
 which, together with \dref{2020819937}, leads to
 $K_{\Lambda}e^{A\tau } S_{\tau  }f \in U$.
 The proof is complete.


\end{proof}

\begin{lemma}\label{lm202061804}
Suppose that  $A$ is the   generator of a  $C_0$-semigroup
 $e^{A t}$  on $Z$,
  $B\in \mathcal{L}(U,[D(A^*)]')$ is admissible for
 $e^{A t}$ and $K\in \mathcal{L}( D(A ),U )$  is admissible for
 $e^{A t}$.
Then, for any $\tau  >0$,    $K\in \mathcal{L}( D(A ),U )$ stabilizes system $(A ,B)$  exponentially   if and only if
   $K e^{ A \tau  }\in \mathcal{L}( D(A ),U )$ stabilizes system  $(A, e^{ -\tilde{A}\tau   }B)$ exponentially.

\end{lemma}
\begin{proof}
By Lemma \ref{lm2020632124} in Appendix,
 both $e^{ -\tilde{A}\tau   }B$    and   $K e^{ A \tau  }$ are  admissible for $e^{At}$.
 For any $\lambda\in \rho(A)$, a simple computation shows that
 \begin{equation} \label{202061817}
  K_{ \Lambda} (\lambda-\tilde{A} ) ^{-1}B=
   K_{\Lambda}\left[ e^{ A \tau  } (\lambda-\tilde{A} )^{-1}e^{-\tilde{A} \tau  }\right]B=
   (K e^{ A \tau  })_{\Lambda}(\lambda-\tilde{A} )^{-1}e^{-\tilde{A} \tau  }B,
 \end{equation}
 where $(K e^{ A \tau  })_{\Lambda}$ is the $\Lambda$-extension of $K e^{ A \tau  }$ with respect to $A$.
 By  Proposition \ref{Regular},   $(A ,B,K)$ is a regular triple if and only if
$(A ,  e^{-\tilde{A} \tau  }B, K e^{ A \tau  })$ is a regular triple.
 Moreover,  $I$   is an admissible feedback operator for    $K_{ \Lambda}(\lambda-\tilde{A} )^{-1}B$
  is equivalent to that
  $I$   is an admissible feedback operator for    $(K e^{ A \tau  })_{\Lambda}(\lambda-\tilde{A} )^{-1}e^{-\tilde{A} \tau  }B$.
Since $e^{ \tilde{A} \tau  }\in \mathcal{L}(Z)$ and
\begin{equation} \label{202061816}
\left(A  + e^{-\tilde{A} \tau  }B K_{\Lambda} e^{  {A} \tau  }\right)z=
\left(A + e^{-\tilde{A} \tau  }B K_{\Lambda} e^{ \tilde{A} \tau  }\right)z= e^{-\tilde{A}\tau   }(A+ BK_{ \Lambda})
e^{ \tilde{A} \tau  }z,\ \ \forall\ z\in Z,
 \end{equation}
  $A +e^{-\tilde{A} \tau  }B K_{\Lambda} e^{  {A} \tau  }$ is exponentially stable in $Z$ if and only if
  $A +BK_{ \Lambda}$ is exponentially stable in $Z$.
Finally,    the proof is completed  by Definition \ref{De20204202110}.
\end{proof}

%
%
%

\section{Input delays compensator design}\label{Actuator}
This section is devoted to the  input  delay compensation.
Let $Z$ and $U$ be Hilbert spaces.
Suppose that the operator $A: D(A)\subset Z\to Z$ generates a $C_0$-semigroup $e^{At}$ on   $Z$ and
$B\in \mathcal{L}(U,[D(A^*)]')$ is admissible for $e^{At}$.
 Consider the following linear   system:
  \begin{equation} \label{20205291857}
 \left.\begin{array}{l}
\dot{z} (t) = Az(t)+B   u(t-\tau),\ \ \tau>0,
\end{array}\right.
\end{equation}
where $z(t)$ is the state  and
 $u:[-\tau,\infty)\to  U $ is the  control  that is delayed
by $\tau$ units of time.
If we let
 \begin{equation} \label{20205292010}
 \phi(x,t)=u(t-  x) ,\ \ x\in[0,\tau ], \ t\geq 0,
 \end{equation}
then,  system \dref{20205291857}
 can be written  as the  delay free form:
 \begin{equation} \label{20205292011}
\left\{\begin{array}{l}
\disp  \dot{z} (t)=A z(t) + B \phi(\tau ,t),\crr
\disp    \phi_t(x,t)+ \phi_x(x,t)=0 \ \ \mbox{in}\ \ U,\ x\in(0,\tau), \crr
\disp \phi(0,t)=u(t).
\end{array}\right.
\end{equation}
We consider system \dref{20205292011} in the state space
$\mathcal{Z}_{\tau  }(U) =Z\times L^2([0,\tau];U)$ with the
inner product
 \begin{equation} \label{20207101553}
\langle (z_1,f_1)^{\top},(z_2,f_2)^{\top}\rangle_{\mathcal{Z}_{\tau  }(U) }=
\langle z_1,z_2\rangle_Z+\langle f_1,f_2\rangle_{L^2([0,\tau];U)},\ \ \forall\
(z_j,f_j)^{\top}\in \mathcal{Z}_{\tau  }(U), \ j=1,2,
\end{equation}
where $\langle \cdot,\cdot\rangle_{L^2([0,\tau];U)}$
is given by \dref{20205302222} with $\alpha=\tau$.
 In terms of     $G_{\tau  }$,  $B_{\tau  }$ and $C_{\tau  }$  defined by \dref{20205292012},
\dref{20205302223}  and  \dref{20205302118} with $\alpha=\tau$, respectively,
   system  \dref{20205292011} can be written  as the  abstract form
\begin{equation} \label{20205292013}
 \left\{\begin{array}{l}
 \dot{z} (t)=\tilde{A} z(t)+B C_{\tau\Lambda}\phi(\cdot,t),\crr
\phi_t(\cdot,t)= \tilde{G}_{\tau}   \phi(\cdot,t)+B_{\tau}u(t).
\end{array}\right.
\end{equation}
 Let $S_{\tau  } :  [H^1([0,\tau  ];U)]'   \to [D(A^*)]'$  be  defined by
\dref{20205312113} and
  \begin{equation} \label{20205302057S}
\disp    \mathbb{S}  \left( z , f \right)^{\top}= \left( z+S_{\tau}  f,\ f \right)^{\top},\ \ \forall\ ( z, f)^\top\in \mathcal{Z}_{\tau  }(U).
\end{equation}
By Lemma \ref{lm20205302114},
  $ \mathbb{S}\in \mathcal{L}(\mathcal{Z}_{\tau  }(U))$ is invertible and its inverse is given by
\begin{equation} \label{20206271218}
   \left.\begin{array}{l}
 \disp    \mathbb{S}^{-1} \left( z , f \right)^{\top}= \left(   z- S_{\tau}  f,
 f\right)^{\top},\ \ \forall \left(z , f \right)^{\top}\in  \mathcal{Z}_{\tau  }(U).
 \end{array}\right.
\end{equation}
Suppose that $(z,\phi) \in C([0,\infty);\mathcal{Z}_{\tau}(U))$ is a solution of system \dref{20205292013}.
Inspired by \cite{FPart1}, we introduce the   transformation
 \begin{equation} \label{20206251921}
 (\tilde{z} (t),\tilde{\phi}(\cdot,t))^{\top}=
\mathbb{S}( {z} (t), {\phi}(\cdot,t))^{\top}.
\end{equation}
By \dref{20205302207} and   \dref{20206271839}, the transformation
 \dref{20206251921}   converts
    system \dref{20205292013}     into
 \begin{equation} \label{202061959}
 \left\{\begin{array}{l}
 \dot{\tilde{z}} (t)=\tilde{A} \tilde{z}(t)+  e^{-\tilde{A}\tau}B u(t),\crr
\tilde{\phi}_t(\cdot,t)= \tilde{G}_{\tau}   \tilde{\phi}(\cdot,t)+B_{\tau}u(t),
\end{array}\right.
\end{equation}
provided
 ${\phi}(\cdot,t)\in G_{B_{\tau}}$.
Since  $G_{\tau}$ is already stable, the stabilization of
system
\dref{202061959}  amounts to the stabilization of  the pair $(A, e^{ -\tilde{A}\tau }B)$.
 By Lemma \ref{lm202061804}, the stabilizer of \dref{202061959} can be designed as
 \begin{equation} \label{202061953}
 u(t)= K_{\Lambda} e^{  {A} \tau}\tilde{z}(t),\ \ t\geq0,
 \end{equation}
 where   $K\in \mathcal{L}( D(A ),U )$ stabilizes system $(A ,B)$ exponentially.
 Under the feedback \dref{202061953},  we get the closed-loop system of \dref{202061959}
\begin{equation} \label{2020611004}
 \left\{\begin{array}{l}
 \dot{\tilde{z}} (t)=\tilde{A} \tilde{z}(t)+  e^{-\tilde{A}\tau}B K_{\Lambda} e^{  {A} \tau}\tilde{z}(t),\crr
\tilde{\phi}_t(\cdot,t)= \tilde{G}_{\tau} \tilde{  \phi}(\cdot,t)+B_{\tau}K_{\Lambda} e^{  {A} \tau}\tilde{z}(t).
\end{array}\right.
\end{equation}
This transformed system can be written abstractly as
 \begin{equation} \label{20206271127}
 \frac{d}{dt} (\tilde{z}(t),\tilde{\phi}(\cdot,t))^{\top}=
 \mathscr{A}_{\mathbb{S}} (\tilde{z}(t),\tilde{\phi}(\cdot,t))^{\top},
\end{equation}
where $\mathscr{A}_{\mathbb{S}} $ is given by
\begin{equation} \label{20206271128}
\mathscr{A}_{\mathbb{S}} =\begin{pmatrix}
             \tilde{A}+ e^{-\tilde{A}\tau}B K_{\Lambda} e^{  {A} \tau}&0 \\
           B_{\tau}K_{\Lambda} e^{  {A} \tau} & \tilde{G}_{\tau}
            \end{pmatrix}
\end{equation}
with
\begin{equation} \label{20206271129}
D(\mathscr{A}_{\mathbb{S}} )=\left\{
\begin{pmatrix}
 z\\
f
\end{pmatrix}
 \in \mathcal{Z }_{\tau}(U)\ \Big{|}\
 \begin{array}{l}
\disp \tilde{A}z+  e^{-\tilde{A}\tau}B K_{\Lambda} e^{  {A} \tau}z\in Z \\ \disp \tilde{G}_{\tau}f+B_{\tau}K_{\Lambda} e^{  {A} \tau} z\in L^2([0,\tau];U)
\end{array}  \right\}.
\end{equation}
 Combining \dref{20206251921}, \dref{202061953} and \dref{20205302131}, the
stabilizer of the original system  \dref{20205292013}  is
 \begin{equation} \label{2020611019}
 \begin{array}{rl}
 u(t)&\disp=\left( K_{\Lambda} e^{  {A} \tau},0\right) (\tilde{z} (t),\tilde{\phi}(\cdot,t))^{\top}
 =\left( K_{\Lambda} e^{  {A} \tau},0\right)\mathbb{S} ( {z} (t), {\phi}(\cdot,t))^{\top}\crr
 &\disp=
  K_{\Lambda} e^{  {A} \tau}\left[z(t)+  S_{\tau} \phi (\cdot,t)\right]
   \disp=K_{\Lambda} e^{  {A} \tau}z(t) +K_{\Lambda}   \int_0^{\tau} e^{\tilde{A}  x }B \phi (x,t)dx,
 \end{array}
 \end{equation}
under which
the closed-loop system of  \dref{20205292011} is:
\begin{equation} \label{2020611024}
 \left\{\begin{array}{l}
 \dot{z} (t)=A z(t)+B  \phi(\tau,t),\crr
\disp    \phi_t(x,t)+ \phi_x(x,t)=0 \ \ \mbox{in}\ \ U,\ x\in(0,\tau), \crr
\disp \phi(0,t)= K_{\Lambda}   \int_0^{\tau} e^{\tilde{A}  x }B \phi (x,t)dx+K_{\Lambda} e^{  {A} \tau}z(t).
\end{array}\right.
\end{equation}
Define
\begin{equation} \label{20206271137}
\mathscr{A}=\begin{pmatrix}
             \tilde{A}&BC_{\tau\Lambda} \\
             B_{\tau}K_{\Lambda} e^{  {A} \tau} & \tilde{G}_{\tau} +B_{\tau}K_{\Lambda} e^{  {A} \tau}
S_{\tau}
            \end{pmatrix}
\end{equation}
with
\begin{equation} \label{20206271138}
D(\mathscr{A})=\left\{
\begin{pmatrix}
  z \\
  f
\end{pmatrix}
 \in \mathcal{Z}_{\tau}(U)\ {\Large \Big|}\
 \begin{array}{l}
   \tilde{A}z+BC_{\tau\Lambda}f\in Z \\
   \tilde{G}_{\tau}f  +B_{\tau}K_{\Lambda} e^{  {A} \tau}\left(
S_{\tau} f+  z\right)\in  L^2([0,\tau];U)
 \end{array}
  \right   \}.
\end{equation}
Then,  the  closed-loop system   \dref{2020611024} can be written abstractly as
\begin{equation} \label{20206271139}
\frac{d}{dt}(z(t),\phi(\cdot,t))^{\top}=\mathscr{A} (z(t),\phi(\cdot,t))^{\top},\ \ t\geq0.
\end{equation}

\begin{theorem}\label{Th2020611038}
Let  $G_{\tau  }$,  $B_{\tau  }$ and $C_{\tau  }$ be  given  by \dref{20205292012},
\dref{20205302223}  and  \dref{20205302118} with $\alpha=\tau>0$, respectively.
 Suppose that  $S_{\tau}$ is given  by \dref{20205312113}  and   $K\in \mathcal{L}( D(A ),U )$ stabilizes system  $(A ,B)$  exponentially.
Then,
      the operator $\mathscr{A}$   defined by \dref{20206271137} generates an exponentially stable $C_0$-semigroup
      $e^{\mathscr{A}t}$ on  $ \mathcal{Z}_{\tau}(U)  $.
    As a result,  for any $(z(0),\phi(\cdot,0))^{\top}\in  \mathcal{Z}_{\tau}(U)  $, system \dref{2020611024} admits a unique solution $(z,\phi)^{\top}\in C([0,\infty); \mathcal{Z}_{\tau}(U)  )$
that decays to zero exponentially in $ \mathcal{Z}_{\tau}(U)  $ as $t\to\infty$.

\end{theorem}
\begin{proof}
By Lemma  \ref{lm20205302114},
   the operator $S_{\tau}$    satisfies   \dref{20205302131}, \dref{20205302207} and
\dref{20206271839}.
 We first claim that
$\mathscr{A}$ is similar to $\mathscr{A}_{\mathbb{S}} $,
  i.e.,
 \begin{equation} \label{20206271213}
 \left.\begin{array}{l}
\disp \mathbb{S}\mathscr{A} \mathbb{S} ^{-1}=\mathscr{A}_{\mathbb{S}} \ \ \mbox{and}\ \  D(\mathscr{A}_{\mathbb{S}} ) = \mathbb{S}D(\mathscr{A}),
\end{array}\right.
\end{equation}
where $\mathbb{S}$ is given by \dref{20205302057S}.


For any $(z,f)^{\top}\in D(\mathscr{A}_{\mathbb{S}} )$,
\dref{20206271129} and \dref{202062291621} imply  that $f\in G_{B_{\tau}}$.
Moreover,
 $K_{\Lambda} e^{  {A} \tau}  S_{\tau}f\in  U$ due to Lemma
\ref{lm20205302114}. Hence,  it follows from   \dref{20206271218}
and \dref{20206271129} that
 \begin{equation} \label{20208191000}
(\tilde{G}_{\tau} +B_{\tau}K_{\Lambda} e^{  {A} \tau}S_{\tau})f+B_{\tau}K_{\Lambda} e^{  {A} \tau} (z-S_{\tau}f)
 =   \tilde{G}_{\tau} f+B_{\tau}K_{\Lambda} e^{  {A} \tau} z\in  L^2([0,\tau];U).
 \end{equation}
 Combine  \dref{20206271839},  \dref{20205302207},  \dref{20206271129}    and the fact
 $S_{\tau  }\in \mathcal{L}(  L^2([0,\tau ];U) , Z)$
   to get
 \begin{equation} \label{2020630812}
\left.\begin{array}{ll}
\disp \tilde{A}(z-S_{\tau}f)+BC_{\tau}f &\disp =
    \tilde{A} z  - S_{\tau}\tilde{G}_{\tau}f  =
    \tilde{A} z+e^{-\tilde{A}\tau}BK_{\Lambda}e^{A\tau}z-  S_{\tau}B_{\tau}K_{\Lambda}e^{A\tau}z
    - S_{\tau}\tilde{G}_{\tau}f\crr
    &\disp =
    (    \tilde{A} z+e^{-\tilde{A}\tau}BK_{\Lambda}e^{A\tau}z)-
     S_{\tau}(\tilde{G}_{\tau} f+B_{\tau}K_{\Lambda} e^{  {A} \tau} z)\in  Z,
\end{array}\right.
 \end{equation}
   which, together with \dref{20208191000}, \dref{20206271138} and \dref{20206271218}, yields
    $\mathbb{S}^{-1}(z,f) ^{\top} \in  D(\mathscr{A} ) $. Hence $D(\mathscr{A}_{\mathbb{S}})\subset \mathbb{S}D(\mathscr{A} )$ due to the arbitrariness of
$(f,z)^ {\top} \in   D(\mathscr{A}_{\mathbb{S}} ) $.

  On the other hand, for any $(z,f)^{\top}\in D(\mathscr{A}  )$,
  \dref{20206271138} and \dref{202062291621} imply  that $f\in G_{B_{\tau}}$ and
  \begin{equation} \label{20208191036}
\tilde{G}_{\tau} f+B_{\tau}K_{\Lambda} e^{  {A} \tau}\left(  S_{\tau }f
+  z\right)\in  L^2([0,\tau];U) .
 \end{equation}
  It follows from    \dref{20205302207},  \dref{20206271839}, \dref{20208191036}  and the fact
 $S_{\tau  }\in \mathcal{L}(  L^2([0,\tau ];U) , Z)$
   that
  \begin{equation} \label{2020630916}
\left.\begin{array}{l}
\disp \tilde{A}(z+S_{\tau}f)+ e^{-\tilde{A}\tau}BK_{\Lambda}e^{A\tau}(z+S_{\tau}f)  =
   \tilde{A}z+S_{\tau}\tilde{G}_{\tau}f+BC_{\tau}f+ S_{\tau}B_{\tau}K_{\Lambda}e^{A\tau}(z+S_{\tau}f)
     \crr
    \hspace{2cm}\disp = ( \tilde{A}z +BC_{\tau}f)+
    S_{\tau}\left[\tilde{G}_{\tau} f+B_{\tau}K_{\Lambda} e^{  {A} \tau}\left(  S_{\tau }f
+  z\right)\right]
    \in  Z.
\end{array}\right.
 \end{equation}
 We combine \dref{20208191036}, \dref{2020630916}, \dref{20205302057S} and \dref{20206271129} to get
   $  \mathbb{S} (z,f)^{\top}\subset  D(\mathscr{A}_{ \mathbb{S}} )$  and hence $  \mathbb{S} D(\mathscr{A} )\subset  D(\mathscr{A}_{ \mathbb{S}} )$.
   In summary, we arrive at
$  \mathbb{S} D(\mathscr{A} ) =  D(\mathscr{A}_{ \mathbb{S}} )$.
Moreover, for any
$(z,f)^{\top}\in D(\mathscr{A}_{\mathbb{S}})$, it follows from \dref{20206271129} and  \dref{202062291621} that
$f\in G_{B_{\tau}}$.
By virtue of \dref{20206271839}, a simple computation shows that
 $\mathbb{S}\mathscr{A} \mathbb{S} ^{-1}(z,f)^{\top}=\mathscr{A}_{\mathbb{S}}(z,f)^{\top}$ for   any
$(z,f)^{\top}\in D(\mathscr{A}_{\mathbb{S}})$. Hence, the similarity \dref{20206271213} holds.

Since   $K\in \mathcal{L}( D(A ),U )$ stabilizes $(A ,B)$ exponentially, it follows from  Lemma \ref{lm202061804}
that  $K e^{ A \tau  }\in \mathcal{L}( D(A ),U )$ stabilizes $(A, e^{ -\tilde{A}\tau   }B)$ exponentially. In particular,   $A +e^{-\tilde{A} \tau  }B K_{\Lambda} e^{  {A} \tau  }$
generates an exponentially   stable $C_0$-semigroup  $e^{(A +e^{-\tilde{A} \tau  }B K_{\Lambda} e^{  {A} \tau  })t}$ on $Z$ and  $K e^{ A \tau  }$ is admissible for
$e^{(A +e^{-\tilde{A} \tau  }B K_{\Lambda} e^{  {A} \tau  })t}$.
By   Lemma \ref{Lm201912031723} in Appendix,    the operator $\mathscr{A}_{\mathbb{S}}$ generates an exponentially stable
   $C_0$-semigroup  $e^{\mathscr{A}_{\mathbb{S}} t}$ on $ \mathcal{Z}_{\tau}(U)  $.
   Owing to the similarity of $\mathscr{A}_{\mathbb{S}}$ and $\mathscr{A}$,
    the operator $\mathscr{A} $ generates an exponentially stable
   $C_0$-semigroup  $e^{\mathscr{A}  t}$ on $ \mathcal{Z}_{\tau}(U)  $ as well.
 The proof is complete.
\end{proof}


\begin{remark}\label{Re2020716940}
When $A$ is a  matrix,  it follows from \dref{20205292010} that the  controller \dref{2020611019}
takes form:
 \begin{equation} \label{2020716942}
 \begin{array}{l}
 u(t)  \disp
 =K   \int_{t-\tau}^{t} e^{A (t-s) }B u (s)dx+K  e^{ A \tau}z(t),\ \ t\geq\tau,
 \end{array}
 \end{equation}
which is the same as
those
obtained from   the spectrum
assignment  approach in \cite{Kwon1980delay}, the ``reduction approach" in \cite{Artstein1982} and the PDE backstepping method in \cite{Krstic2008delay}.
  We  point out that
  the Lyapunov function has not been  used in
  the  stability analysis  of our    method. This avoids the difficulty about the Lyapunov-based technique for stabilization of PDEs with  delay.
  Another advantage of the proposed approach is that
 we never need the target system as that by the backstepping
approach. This avoids the possibility that when the target system  is not chosen properly, there is no state
feedback control and even if the target system is good enough, there is difficulty in solving PDE
kernel equation for the backstepping transformation.
  \end{remark}

\section{Output delays and
   Sylvester equation} \label{OupputSyl}

In this and next sections, we consider the   output   delay compensation,  which
is  the most common
dynamic phenomena arising in control engineering practice.
 Consider the following system in the state space $Z$, input space $U $ and the output space $Y $: \begin{equation} \label{wxh2020321102}
 \left.\begin{array}{l}
\dot{z} (t) = A z (t)+B  u(t),\ \
y(t)=C_{\Lambda} z (t-\mu ),\ \ \mu>0,
\end{array}\right.
\end{equation}
 where
$A : D(A )\subset Z \to Z $
is the system operator,
$B  \in \mathcal{L}(U , [D(A ^*)]')$ is the control operator,  $C \in \mathcal{L}(D(A) ,Y )$ is the  observation operator,
 $u (t)  $ is  the  control input,  and   $y(t)$ is the    measurement that is delayed
by $\mu $ units of time. Let $\psi(x,t)=C_{\Lambda}z(t-x)$ for $x\in[0,\mu ]$ and $t\geq \mu$.
Then,  system \dref{wxh2020321102} can be written as
\begin{equation} \label{wxh2020321104}
\left\{\begin{array}{l}
\disp  \dot{z} (t)=A z (t) + B u(t),\crr
\disp     \psi_t(x,t)+ \psi_x(x,t)=0\  \ \mbox{in}\ Y ,\ \ x\in[0,\mu ], \crr
\disp  \psi(0,t)=C_{\Lambda}z (t),\crr
\disp y(t)= \psi(\mu ,t).
\end{array}\right.
\end{equation}
We consider system \dref{wxh2020321104} in state space $\mathcal{Z}_{\mu}(Y)= Z\times  L^2([0,\mu];Y)$.  The inner product of $\mathcal{Z}_{\mu}(Y)$ is given by
\dref{20207101553} with $\alpha=\mu$ and $U=Y$.
 In terms of the operators $G_{\mu}$, $B_{\mu}$ and $C_\mu$, which are given by
 \dref{20205292012}, \dref{20205302223}  and  \dref{20205302118} with $\alpha=\mu$ and $U=Y$, respectively,
system  \dref{wxh2020321104} can be written  as
\begin{equation} \label{wxh2020321120}
 \left\{\begin{array}{l}
  \dot{z} (t)=\tilde{A} z (t)+B u(t),\crr
  \psi_t(\cdot,t)= \tilde{G}_{\mu} \psi(\cdot,t)+B_{\mu}C_{\Lambda}z(t),\crr
  y(t)=C_{\mu\Lambda}\psi(\cdot,t) .
\end{array}\right.
\end{equation}
The following Theorem guarantees  that    the mapping from each   initial data and control input signal to the   state and the output observation signal is continuous.

\begin{theorem}\label{th2020613745}
Let  $G_{\mu}$, $B_{\mu}$ and $C_\mu$ be  given by
 \dref{20205292012}, \dref{20205302223}  and  \dref{20205302118}  with $\alpha=\mu $ and   $U=Y$, respectively.
Suppose that $(A,B,C)$ is  a well-posed linear system in the sense of   Salamon in \cite{Salamon1987}. Then,
 system \dref{wxh2020321120}   is also well-posed: For any $(z(0),\psi(\cdot,0))^{\top}\in \mathcal{Z}_{\mu}(Y)$ and $u\in L^2_{\rm loc}([0,\infty);U)$, system \dref{wxh2020321120} admits a
 unique solution $(z ,\psi )^{\top}\in C([0,\infty);\mathcal{Z}_{\mu}(Y))$ that satisfies, for any
$ T>0$, there exists a positive constant
 $C_T$
  such that
 \begin{equation} \label{2020613754}
\int_0^T\|y(s)\|_Y^2ds+\|(z(T),\psi(\cdot,T))^{\top}\|_{ \mathcal{Z}_{\mu}(Y) }\leq
C_T\left[\int_0^T\|u(s)\|_U^2ds+\|(z(0),\psi(\cdot,0))^{\top}\|_{ \mathcal{Z}_{\mu}(Y) }\right].
\end{equation}
 \end{theorem}
\begin{proof}
  Since $z$-subsystem  is independent of $\psi$-subsystem,
  the solution of \dref{wxh2020321120} can be expressed explicitly:
  \begin{equation} \label{2020613803}
z(t)=e^{At}z(0)+\int_0^te^{A(t-s)}Bu(s)ds,\ \ \psi(x,t)=\left\{\begin{array}{ll}
  C_{ \Lambda} z(t-x),&  t-x\geq0,\\
  \psi(x-t,0),&t-x<0,
\end{array}\right.
\end{equation}
where $x\in[0,\mu]$. Moreover,
\begin{equation} \label{2020613816}
y(t)= C_{\mu\Lambda}\psi(\cdot,t)=\psi(\mu,t)=\left\{\begin{array}{ll}
  C _{ \Lambda}z(t-\mu),&  t-\mu\geq0,\\
  \psi(\mu-t,0),&t-\mu<0
\end{array}\right.
\end{equation}
and
\begin{equation} \label{2020711953}
\begin{array}{ll}
\disp \|\psi(\cdot,t)\|_{L^2([0,\mu];Y)}^2
&\disp =\left\{\begin{array}{ll}
  \disp \int_0^t\| C _{ \Lambda}z(t-x)\|_Y^2dx+\int_t^\mu\|\psi(x-t,0)\|_Y^2dx,&  0\leq t<\mu \crr
  \disp \int_0^\mu\| C _{ \Lambda}z(t-x)\|_Y^2dx,&t\geq\mu
\end{array}\right.\crr
&\disp =\left\{\begin{array}{ll}
  \disp \int_0^t\| C _{ \Lambda}z( x)\|_Y^2dx+\int_0^{\mu-t} \|\psi(x ,0)\|_Y^2dx,&  0\leq t<\mu\crr
  \disp \int_{t-\mu}^t\| C _{ \Lambda}z( x)\|_Y^2dx,&t\geq\mu .
\end{array}\right.
\end{array}
\end{equation}
 Since  $(A,B,C)$ is    well-posed,
  for any $t>0$, there exists a $C_t>0$ such that
\begin{equation} \label{2020613804}
\int_{0}^{t}\|C_{ \Lambda} z(s)\|_Y^2ds+\|z(t)\|_Z^2\leq C_t\left[\|z(0)\|_Z^2+
\int_{0}^{t}\|u(s)\|_U^2ds \right],
\end{equation}
 which, together with  \dref{2020613803},  \dref{2020613816} and \dref{2020711953}, leads to \dref{2020613754} easily.
  The proof is complete.
\end{proof}

Let  $(A,C)$ be an observation   system with the state space $Z$ and output space $Y$.
Suppose that  $A$  generates a  $C_0$-semigroup $e^{At}$ on $Z$ and
$C\in \mathcal{ L}(D(A),Y)$ is admissible for $e^{At} $.  As defined in
\cite[Definition 4.1.12, p.154]{Curtain1995},   for   any $\mu>0$,
  the observability map  of system $(-A,-C)$  is
 \begin{equation} \label{20206141839}
\begin{array}{ll}
\Psi_{\mu} : &\disp Z\to L^2([0,\mu];Y)\crr
          &\disp    z\to -C_{ \Lambda}e^{  -{A} x}z ,\ \   x\in [0,\mu ],\ \ \forall\ z\in Z.
\end{array}
\end{equation}
Since $C$ is admissible for the $C_0$-semigroup $e^{-At}$,
  $\Psi_{\mu} \in\mathcal{L} (Z,  L^2([0,\mu];Y) )$.
For any     $f\in  D(G_{\mu}^*  )\subset H^1([0,\mu];Y)$,
    it follows from   \dref{20191121602},  \dref{20205302138},  \dref{20205302223}  and the fact
$ [H^1([0,\mu];Y)]'\subset[D(G_{\mu}^*)]'$  that
  \begin{equation} \label{202071817}
  \begin{array}{ll}
\disp \left\langle \tilde{G}_{\mu}\Psi_{\mu}z , f\right\rangle_{[D(G_{\mu}^*)]',D(G_{\mu}^*)}
&\disp =
 \left\langle - C_{\Lambda }e^{-A \cdot}z ,  {G}_{\mu} ^* f\right\rangle_{ L^2([0,\mu];Y) }\crr
&\hspace{-4cm}\disp =
 -\int_0^{\mu }\left\langle C_{\Lambda }e^{-A \sigma}z,\frac{d}{d\sigma}{f} (\sigma) \right\rangle_{Y}d\sigma
 \disp =
  \langle Cz, f(0) \rangle_{Y}-
 \int_0^{\mu }\langle C_{ \Lambda }e^{-A \sigma}Az,f(\sigma) \rangle_{Y}d\sigma
 \crr
 &\hspace{-4cm} \disp =
  \langle B_{\mu}Cz, f  \rangle_{[D(G_{\mu}^*)]',D(G_{\mu}^*)}+
  \langle \Psi_{\mu}A z,f  \rangle_{ L^2([0,\mu];Y) },\ \ \forall\  z\in D(A) .
 \end{array}
 \end{equation}
Owing to the arbitrariness of $f \in  D(G_{\mu}^*  )$,  \dref{202071817} implies that
 the following Sylvester equation holds in $[D(G_{\mu}^*)]'$:
\begin{equation} \label{202070812}
 \tilde{G}_{\mu}\Psi_{\mu}z-\Psi_{\mu}Az= B_{\mu}C_{\Lambda}z ,  \ \ \forall\ z\in D(A).
 \end{equation}

Suppose that   $(A,F_1,C)$ is   a  linear system with the state space $Z$, input space $Y$ and the output space $Y$.
We define a subspace of  $Z$ by
 \begin{equation} \label{20206301054}
 Z_{F_1} =\left\{z\in Z \ |\ \tilde{A}z+F_1y\in Z, y\in Y\right\}.
 \end{equation}
As in  \cite[Section 2.2]{Salamon1987},
  $Z_{F_1}$ with   inner product
  \begin{equation} \label{2020711135}
\| z\|_{ Z_{F_1}}^2= \| z\|_{Z}^2+
\|y_z\|_Y^2+\| \tilde{A}z+F_1y_z \|^2_{Z}
\end{equation}
  is   a Hilbert space, where $y_z \in Y$ such that
  $ \tilde{A}z+  F_1 y_z
  \in Z $.


\begin{lemma}\label{lm2020516938}
Let $Z$ and $Y$ be Hilbert spaces.
Suppose that $A$   generates    a $C_0$-semigroup $e^{At}$  on
$Z$, $C\in  \mathcal{L}(D(A),Y )$ is admissible for $e^{At}$ and  $F_1\in  \mathcal{L}(Y,[D(A^*)]' )$.
Then, $Z_{F_1}$ defined by  \dref{20206301054} satisfies
 \begin{equation} \label{20206301054**}
 Z_{F_1}   =D(A )+ (\lambda-\tilde{A} )^{-1}F_1Y,\ \ \lambda\in \rho(A ).
 \end{equation}
  Suppose that $ Z_{F_1}\subset D(C_{\Lambda})$  and $G_{\mu }$,  $B_{\mu  }$ and $C_{\mu  }$ are  defined by \dref{20205292012},
\dref{20205302223}  and  \dref{20205302118} with $\alpha=\mu$ and $U=Y$, respectively.
Define
the operator $P_{\mu}: (Z+F_1Y)\subset [D(A^*)]' \to [D(G_{\mu}^*)]'$ by
  \begin{equation} \label{202071733}
  P_{\mu} z =
     \left[B_{\mu}C_{\Lambda} +(\lambda-\tilde{G}_{\mu}) \Psi_{\mu}\right](\lambda-\tilde{A} )^{-1}z,\ \ \forall z\in (Z+F_1Y),
     \end{equation}
     where $\lambda\in \rho(A)$.
  Then,   the following assertions hold  true:

(i)~ $ P_{\mu} $  is independent of $\lambda$ and  is an extension of $\Psi_{\mu}$, i.e.,
\begin{equation} \label{202071808}
   P_{\mu}z
      = \Psi_{\mu}z,\ \ \forall\ \ z\in Z;
  \end{equation}

(ii)~ $P_{\mu}$ is satisfied by the following Sylvester equation on $Z_{F_1}$:
   \begin{equation} \label{2020711041}
 \tilde{G}_{\mu}P_{\mu}z-P_{\mu}\tilde{A}z= B_{\mu}C_{\Lambda}z ,  \ \ \forall\ z\in Z_{F_1};
  \end{equation}

(iii)~  $P_{\mu}$ and $C_{\mu\Lambda}$ satisfy:
\begin{equation} \label{202071900}
  C_{\mu\Lambda}P_{\mu}= -C_{\Lambda}e^{-A\mu }\in \mathcal{L}(D(A),Y ).
   \end{equation}


 \end{lemma}
\begin{proof}
For any $z\in Z_{F_1}\subset Z$, there exists a $y  \in Y$ such that
 $ \tilde{A} z+F_1y \in Z  $ and hence $ (\lambda-\tilde{A}) z-F_1y \in Z  $ for any $\lambda\in\rho(A)$.
As  a result,   $z-(\lambda-\tilde{A})^{-1} F_1y  \in D(A) $  and $z\in D(A)+(\lambda-\tilde{A})^{-1} F_1Y$.
So $Z_{F_1}\subset D(A)+(\lambda-\tilde{A})^{-1} F_1Y$.
For any $z=z_1+(\lambda-\tilde{A})^{-1} F_1y_z \in D(A)+(\lambda-\tilde{A})^{-1} F_1Y$, where $z_1\in D(A)$ and $y_z\in Y$,
a simple computation shows that
$ (\lambda-\tilde{A})z+F_1(-y_z)=  (\lambda-\tilde{A} )z_1 \in
 Z$ and hence $  \tilde{A} z+F_1 y_z   \in
 Z$. By \dref{20206301054},
   $z\in  Z_{F_1}  $  and hence $D(A)+(\lambda-\tilde{A})^{-1} F_1Y \subset Z_{F_1} $. Therefore, \dref{20206301054**} holds.

Proof of (i).
 Since  $ (\lambda-\tilde{A} )^{-1}z=(\lambda- {A} )^{-1}z\in D(A)$   for any $z\in Z$,
  it  follows from \dref{202070812} that
   \begin{equation} \label{202071808***}
   \begin{array}{ll}
  P_{\mu}z &\disp  =
      \left[B_{\mu}C_{\Lambda} +(\lambda-\tilde{G}_{\mu}) \Psi_{\mu}\right](\lambda- {A} )^{-1}z\crr
      &\disp  =
      - \Psi_{\mu}A    (\lambda- {A} )^{-1}z+\lambda \Psi_{\mu} (\lambda- {A} )^{-1}z\crr
  &\disp  =
        \Psi_{\mu}(\lambda-A)    (\lambda- {A} )^{-1}z= \Psi_{\mu}z,\ \ \forall\ \ z\in Z.
      \end{array}
  \end{equation}
Hence,   $P_{\mu} $ is an extension of $\Psi_{\mu}$.

 Proof of (ii). For any $z_{F_1}\in Z_{F_1}$, by \dref{20206301054**}, there exist   $z\in  D(A   )$ and
 $ y\in Y$ such that $z_{F_1}=z+ (\lambda-\tilde{A} )^{-1}F_1y$ for some $\lambda\in \rho(A )$.
 Thanks to \dref{202070812} and \dref{202071808},  it suffices  to prove
\begin{equation} \label{20207181138}
\disp   \tilde{G}_{\mu}P_{\mu}[(\lambda-\tilde{A} )^{-1}F_1 y]  -B_{\mu}C_{\Lambda}[(\lambda-\tilde{A} )^{-1}F_1 y]
 =   P_{\mu } \tilde{A} [ (\lambda-\tilde{A})^{-1}F_1 y].
  \end{equation}
Actually, it follows \dref{202071733} and \dref{202071808} that
  \begin{equation} \label{2020711104}
  P_{\mu} F_1y =
    B_{\mu}C_{\Lambda}(\lambda-\tilde{A})^{-1}F_1y +\lambda P_{\mu}(\lambda-\tilde{A})^{-1}F_1y
     -\tilde{G}_{\mu}P_{\mu}(\lambda-\tilde{A} )^{-1}F_1y,
\end{equation}
which yields
  \begin{equation} \label{20205151747}
  \begin{array}{l}
\disp   \tilde{G}_{\mu}P_{\mu}[(\lambda-\tilde{A} )^{-1}F_1 y]  -B_{\mu}C_{\Lambda}[(\lambda-\tilde{A} )^{-1}F_1 y]
 =
 -P_{\mu }F_1 y+ \lambda P_{\mu}(\lambda-\tilde{A})^{-1}F_1 y\crr
 =
 -P_{\mu }(\lambda-\tilde{A}) (\lambda-\tilde{A})^{-1}F_1 y+ \lambda P_{\mu}(\lambda-\tilde{A})^{-1}F_1 y
 =   P_{\mu } \tilde{A} [ (\lambda-\tilde{A})^{-1}F_1 y].
 \end{array}
 \end{equation}
Hence, \dref{2020711041} can be obtained by  \dref{202070812},  \dref{20205151747} and the fact
$z_{F_1}=z+ (\lambda-\tilde{A} )^{-1}F_1y$
easily.

Proof of (iii). It follows from  \dref{202070812} that
 \begin{equation} \label{2020711452}
 \tilde{G}_{\mu}\Psi_{\mu}z-B_{\mu}C_{\Lambda}z
 =\Psi_{\mu}Az\in L^2([0,\mu];Y) ,  \ \ \forall\ z\in D(A),
 \end{equation}
which, together with \dref{202062291621}, \dref{20206292003}    and \dref{202071808},  yields
$P_{\mu}z=\Psi_{\mu}z\in G_{B_{\mu}}= H^1([0,\mu];Y)=D(C_{\mu})$, where $G_{B_{\mu}}$ is defined by \dref{202062291621} with $\alpha=\mu$. Owing to \dref{20206141839} and \dref{20205302118}, we arrive at
$ C_{\mu\Lambda}P_{\mu}z=C_{\mu\Lambda}\Psi_{\mu}z= -C_{\Lambda}e^{-A\mu }z\in Y $. So
\dref{202071900} holds.
%
%
%
 The proof is complete.
 \end{proof}

\begin{remark}\label{Re2020727931}
 We claim  that  $P_{\mu}$ is independent of the choice of $\lambda$.
 So   the  notation  $P_{\mu}$ that is absent of $\lambda$ does not cause any
confusion.
  Indeed, for any $\lambda_1,\lambda_2\in \rho(A)$ and $\lambda_1\neq\lambda_2$,
   a simple computation shows that
   \begin{equation} \label{202071752}
\left.\begin{array}{l}
  B_{\mu}C_{\Lambda}\left[ (\lambda_1-\tilde{A})^{-1}-  (\lambda_2-\tilde{A})^{-1}\right]
 =     (\lambda_2-\lambda_1) B_{\mu}C_{\Lambda}(\lambda_1-\tilde{A})^{-1}(\lambda_2-\tilde{A})^{-1}
  \end{array}\right.
  \end{equation}
  and
  \begin{equation} \label{202071754}
\left.\begin{array}{l}
 -\tilde{G}_{\mu}\Psi_{\mu}\left[ (\lambda_1-\tilde{A})^{-1}-  (\lambda_2-\tilde{A})^{-1}\right]
 =     -(\lambda_2-\lambda_1)  \tilde{G}_{\mu}\Psi_{\mu}(\lambda_1-\tilde{A})^{-1}(\lambda_2-\tilde{A})^{-1}.
  \end{array}\right.
  \end{equation}
Notice that $(\lambda_1-\tilde{A})^{-1} (\lambda_2-\tilde{A})^{-1}z \in D(A) $ for any $z\in Z+F_1 Y \subset [D(A^*)]'$, it follows from \dref{202070812} that
\begin{equation} \label{202071838}
 \tilde{G}_{\mu}\Psi_{\mu}\hat{z}-\Psi_{\mu}A\hat{z}= B_{\mu}C_{\Lambda}\hat{z} ,  \ \
 \hat{z}=(\lambda_1-\tilde{A})^{-1} (\lambda_2-\tilde{A})^{-1}z\in D(A).
 \end{equation}
Combining \dref{202070812}, \dref{202071808}, \dref{202071752},  \dref{202071754} and \dref{202071838},
for any $z\in Z+F_1 Y $,
we obtain
 \begin{equation} \label{202071736}
\left.\begin{array}{l}
\left[B_{\mu}C_{\Lambda} +(\lambda_1-\tilde{G}_{\mu}) \Psi_{\mu}\right](\lambda_1-\tilde{A} )^{-1}z-
\left[B_{\mu}C_{\Lambda} +(\lambda_2-\tilde{G}_{\mu}) \Psi_{\mu}\right](\lambda_2-\tilde{A} )^{-1}z\crr
=  \disp  (\lambda_2-\lambda_1)\left[ B_{\mu}C_{\Lambda}\hat{z}- \tilde{G}_{\mu}\Psi_{\mu}\hat{z}\right]
+\lambda_1\Psi_{\mu} (\lambda_1-\tilde{A} )^{-1}z-\lambda_2\Psi_{\mu} (\lambda_2-\tilde{A} )^{-1}z\crr
=  \disp  -(\lambda_2-\lambda_1) \Psi_{\mu}A\hat{z}
+\lambda_1\Psi_{\mu} (\lambda_1-\tilde{A} )^{-1}z-\lambda_2\Psi_{\mu} (\lambda_2-\tilde{A} )^{-1}z\crr
=  \disp  - P_{\mu}\tilde{A} \left[(\lambda_1-\tilde{A} )^{-1}z-(\lambda_2-\tilde{A} )^{-1}z \right]
+\lambda_1P_{\mu} (\lambda_1-\tilde{A} )^{-1}z-\lambda_2P_{\mu} (\lambda_2-\tilde{A} )^{-1}z\crr
=  \disp    P_{\mu}(\lambda_1- \tilde{A})  (\lambda_1-\tilde{A} )^{-1}z-   P_{\mu}(\lambda_2- \tilde{A})  (\lambda_2-\tilde{A} )^{-1}z= P_{\mu}z-P_{\mu}z=0.
  \end{array}\right.
  \end{equation}
Therefore, $P_{\mu}$ is independent of the choice of $\lambda$.
\end{remark}

\begin{lemma}\label{lm2020613824}
 Suppose that $A$ is the generator of the $C_0$-semigroup $e^{At}$ acting on
$Z$ and $C\in  \mathcal{L}(D(A),Y )$  is admissible for $e^{At}$. Then for any $\mu> 0$,
 $F\in \mathcal{L}(Y,[D(A^*)]')$ detects system $(A,C)$  exponentially  if and only if
   $ e^{ \tilde{A}\mu }F$ detects system $(A, C_{\Lambda}e^{-A\mu })$ exponentially.

\end{lemma}
\begin{proof}
By Lemma \ref{lm2020632124} in Appendix,
 both
$C_{\Lambda}e^{-A\mu }$   and $e^{ \tilde{A}\mu }F$ are admissible for $e^{At}$.
 Similarly to \dref{202061817},  a simple computation shows that
   \begin{equation} \label{20205161622}
 C_{\Lambda}e^{-A\mu } (\lambda-\tilde{A} )^{-1}e^{ \tilde{A} \mu }F= C_{ \Lambda} (\lambda-\tilde{A} ) ^{-1}F,
 \ \forall\ \lambda\in\rho(A ),
 \end{equation}
 which implies that $(A ,F,C )$ is a regular triple if and only if
$(A ,e^{ \tilde{A} \mu }F, C_{\Lambda}e^{-A\mu })$ is a regular triple.
 Moreover,  $I$   is an admissible feedback operator for    $C_{ \Lambda}(s-\tilde{A} )^{-1}F$
  is equivalent to that
  $I$   is an admissible feedback operator for    $ C_{\Lambda}e^{-A\mu }(\lambda-\tilde{A} )^{-1}e^{ \tilde{A} \mu }F$. Since $e^{ \tilde{A} \mu  }\in \mathcal{L}(Z)$ and
  \begin{equation} \label{20205161634}
A +e^{ \tilde{A} \mu }F C_{\Lambda}e^{-A\mu }= A +e^{ \tilde{A} \mu }F C_{ \Lambda}e^{ -\tilde{A} \mu }
=
A +F C_{ \Lambda} ,
 \end{equation}
  $A+F C_{ \Lambda}$ is exponentially stable if and only if
  $A +e^{-\tilde{A} \mu }F C_{\Lambda}e^{A\mu }$ is exponentially stable.
  The proof is complete due to Definition \ref{De20204191058}.
\end{proof}


\section{ Luenberger-like observer} \label{observer}
In this section, we will design the observer  for  system  \dref{wxh2020321120} and prove the well-posedness.
  We begin with the  following    infinite-dimensional Luenberger-like observer:
\begin{equation} \label{20206112004}
 \left\{\begin{array}{l}
  \dot{\hat{z}} (t)=A \hat{z} (t)-F_1[y(t)-C_{\mu\Lambda}\hat{\psi}(\cdot,t)]+B u(t),\crr
  \hat{\psi}_t(\cdot,t)= G_{\mu} \hat{\psi}(\cdot,t)+B_{\mu}C_{ \Lambda}\hat{z}(t)-F_2[y(t)-C_{\mu\Lambda}\hat{\psi}(\cdot,t)],
  \end{array}\right.
\end{equation}
where $F_1\in \mathcal{L}(Y,[D(A^*)]') $ and $F_2\in \mathcal{L}(Y,[D(G_\mu^*)]') $ are
  tuning  operators to be determined.
If we let the errors be
\begin{equation} \label{20206112014}
\tilde{z} (t)=z(t)-\hat{z} (t),\ \  \tilde{\psi}(\cdot,t)=\psi(\cdot,t)-\hat{\psi}(\cdot,t),
\end{equation}
then they are  governed by
\begin{equation} \label{20206112015}
 \left\{\begin{array}{l}
\disp \dot{\tilde{z}} (t) =A\tilde{z} (t)+F_1C_{\mu\Lambda}\tilde{\psi}(\cdot,t),\crr
\disp  \dot{\tilde{\psi}}_t(\cdot,t) = ( G_{\mu}+F_2C_{\mu\Lambda}) \tilde{\psi}(\cdot,t)+B_{\mu}C_{ \Lambda}\tilde{z} (t).
\end{array}\right.
\end{equation}
Similarly to \dref{20205302057S} and \dref{20206271218}, if $ Z_{F_1}\subset D(C_{\Lambda})$,
 we can define   the transformation
 \begin{equation} \label{20206141610}
 \mathbb{P} \left( z , f \right)^{\top}= \left( z,\ f +P_{\mu}  z\right)^{\top},\ \ \forall\ ( z, f)^\top\in \mathcal{Z}_{\mu}(Y) ,
\end{equation}
where the operator  $P_{\mu} $  is given  by \dref{202071733}.
 It is easy to see that $ \mathbb{P}\in \mathcal{L}( \mathcal{Z}_{\mu}(Y) )$ is invertible and its inverse is given by
\begin{equation} \label{20206301044}
   \left.\begin{array}{l}
 \disp    \mathbb{P}^{-1} \left( z , f \right)^{\top}= \left(   z ,
 f-P_{\mu}  z\right)^{\top},\ \ \forall \ \left(z , f \right)^{\top}\in  \mathcal{Z}_{\mu}(Y) .
 \end{array}\right.
\end{equation}
Let
 \begin{equation} \label{20206301059}
 \left( \check{z}(t) , \check{\psi}(\cdot,t) \right)^{\top}= \mathbb{P} \left( \tilde{z}(t) , \tilde{\psi}(\cdot,t) \right)^{\top},\ \ t\geq0 .
\end{equation}
By \dref{2020711041},   the  transformation \dref{20206301059} can convert  system \dref{20206112015}
into
  \begin{equation} \label{20206301102}
 \left\{\begin{array}{l}
\disp \dot{\check{z}} (t) =({A}-F_1C_{\mu\Lambda}P_{\mu})\check{z} (t)+F_1C_{\mu\Lambda}\check{\psi}(\cdot,t),\crr
\disp  \dot{\check{\psi}}_t(\cdot,t) = (  G_{\mu}+F_2C_{\mu\Lambda}+P_{\mu}F_1C_{\mu\Lambda}) \check{\psi}(\cdot,t)
-(P_{\mu}F_1C_{\mu\Lambda}P_{\mu}+F_2C_{\mu\Lambda}P_{\mu})\check{z} (t),
\end{array}\right.
\end{equation}
provided $\tilde{z}(t)\in Z_{F_1} $.
 Choosing  specially $ F_2=-P_{\mu}F_1$,
system  \dref{20206301102} is reduced to
 \begin{equation} \label{20206301105}
 \left\{\begin{array}{l}
\disp \dot{\check{z}} (t) =({A}-F_1C_{\mu\Lambda}P_{\mu})\check{z} (t)+F_1C_{\mu\Lambda}\check{\psi}(\cdot,t),\crr
\disp  \dot{\check{\psi}}_t(\cdot,t) =    \tilde{G}_{\mu} \check{\psi}(\cdot,t)
,
\end{array}\right.
\end{equation}
which is a simple cascade system and
 can be written as
\begin{equation} \label{20206301116}
\frac{d}{dt}(\check{z} (t),\check{\psi} (\cdot,t))^{\top}
= \mathscr{A}_{\mathbb{P}}(\check{z} (t),\check{\psi} (\cdot,t))^{\top} ,
\end{equation}
where
  \begin{equation} \label{20206301117}
  \left\{\begin{array}{l}
 \disp   \mathscr{A}_{\mathbb{P}}= \begin{pmatrix}
\tilde{A}-F_1C_{\mu\Lambda}P_{\mu} &   F_1C_{\mu\Lambda} \\
0 &   \tilde{G}_{\mu}
\end{pmatrix},\crr
\disp  D(  \mathscr{A}_{\mathbb{P}})=\left\{
\begin{pmatrix}
z\\
\psi
\end{pmatrix}\in  \mathcal{Z}_{\mu}(Y) \ \Big{|}\
\begin{array}{l}
\disp (\tilde{A}-F_1C_{\mu\Lambda}P_{\mu})z+F_1C_{\mu\Lambda}\psi\in Z \\
\disp      \tilde{G}_{\mu} \psi \in  L^2([0,\mu];Y)
\end{array} \right \}  .
\end{array}\right.
\end{equation}
With the setting  $ F_2=-P_{\mu}F_1$,
  the   observer \dref{20206112004}
    is reduced to be
 \begin{equation} \label{20206142114}
 \left\{\begin{array}{l}
\disp \dot{\hat{z}} (t)=\tilde{A} \hat{z} (t)
-F_1[y(t)-C_{\mu\Lambda}\hat{\psi}(\cdot,t)]+B u(t),\crr
\hat{\psi}_t(\cdot,t)= \tilde{G}_{\mu} \hat{\psi}(\cdot,t)+B_{\mu}C_{ \Lambda}\hat{z}(t)
+P_{\mu} F_1
[y(t)-C_{\mu\Lambda}\hat{\psi}(\cdot,t)].
\end{array}\right.
\end{equation}
 System \dref{20206142114} can be written as
\begin{equation} \label{2020613845}
\frac{d}{dt}(\hat{z} (t),\hat{\psi} (\cdot,t))^{\top}
= \mathscr{A}(\hat{z} (t),\hat{\psi} (\cdot,t))^{\top}+ \mathscr{F}y(t) +(B,0)^{\top }u(t),
\end{equation}
where
  \begin{equation} \label{2020613844}
 \disp  {\mathscr{A}}= \begin{pmatrix}
  \tilde{A} &F_1 C_{\mu\Lambda} \\
B_{\mu}C_{ \Lambda} &  \tilde{G}_{\mu}-
P_{\mu}F_1 C_{\mu\Lambda}
\end{pmatrix}\ \ \mbox{and}\ \
{\mathscr{F}}=\begin{pmatrix}
   -F_1     \\
 P_{\mu}F_1
\end{pmatrix}
\end{equation}
with
  \begin{equation} \label{2020613844D}
\disp  D(  \mathscr{A} )=\left\{
\begin{pmatrix}
z\\
\psi
\end{pmatrix}\in  \mathcal{Z}_{\mu}(Y) \ \Big{|}\
\begin{array}{l}
\disp  \tilde{A}z+F_1C_{\mu\Lambda}\psi\in Z \\
\disp    B_{\mu}C_{ \Lambda}z+  (\tilde{G}_{\mu}-P_{\mu}F_1 C_{\mu\Lambda}) \psi \in  L^2([0,\mu];Y)
\end{array} \right \}  .
\end{equation}

\begin{lemma}\label{Lm2020615849}

Let $Z$ and $Y$ be Hilbert spaces.
Suppose that $A$ is a generator of the $C_0$-semigroup $e^{At}$ acting on
$Z$, $C\in  \mathcal{L}(D(A),Y )$ is admissible for $e^{At}$,  $F_1\in  \mathcal{L}(Y,[D(A^*)]' )$
and  $Z_{F_1}$ defined by  \dref{20206301054} satisfies $ Z_{F_1}\subset D(C_{\Lambda})$.
  Suppose that   $G_{\mu }$,  $B_{\mu  }$ and $C_{\mu  }$ are  defined by \dref{20205292012},
\dref{20205302223}  and  \dref{20205302118} with $\alpha=\mu$, respectively.
 Let $ \mathscr{A}$ and $ \mathscr{A}_{\mathbb{P}}$ be given by \dref{2020613844} and \dref{20206301117}, respectively. Then,
 \begin{equation} \label{20206301557}
 \left.\begin{array}{l}
\disp \mathbb{P}\mathscr{A} \mathbb{P} ^{-1}=\mathscr{A}_{\mathbb{P}} \ \ \mbox{and}\ \  D(\mathscr{A}_{\mathbb{P}} ) = \mathbb{P}D(\mathscr{A}),
\end{array}\right.
\end{equation}
  where $\mathbb{P}$ is given by \dref{20206141610}.
\end{lemma}
\begin{proof}
By Lemma \ref{lm2020516938},   the operator
 $P_{\mu}  $   is well defined via \dref{202071733}.
  For any $(z,\psi)^{\top}\in D(\mathscr{A}_{\mathbb{P}} )$,  it follows from \dref{20206301117} and \dref{20206301054} that   $z\in Z_{F_1}$ and
  \begin{equation} \label{20208201003}
    \tilde{A} z     +  F_1 C_{\mu\Lambda} (\psi-P_{\mu}z)\in Z  .
 \end{equation}
  By \dref{20206301117}, \dref{202071808}, \dref{2020711041} and \dref{20206141839},  it follows that \begin{equation} \label{2020630200}
\left.\begin{array}{l}
  \disp   B_{\mu}C_{ \Lambda}z+  (\tilde{G}_{\mu}-P_{\mu}F_1 C_{\mu\Lambda}) (\psi-P_{\mu}z) \disp  =-P_{\mu}\tilde{A}z+   \tilde{G}_{\mu}\psi
  -P_{\mu}F_1 C_{\mu\Lambda}  (\psi-P_{\mu}z)\crr
   \hspace{2cm}\disp  =   \tilde{G}_{\mu}\psi-P_{\mu}\left[\tilde{A} z     +  F_1 C_{\mu\Lambda} (\psi-P_{\mu}z)\right]\crr
   \hspace{2cm}\disp =
\tilde{G}_{\mu}\psi-\Psi_{\mu}\left[\tilde{A} z     +  F_1 C_{\mu\Lambda} (\psi-P_{\mu}z)\right]
\in L^2([0,\mu];Y).
\end{array}\right.
 \end{equation}
 We combine \dref{2020613844D}, \dref{20208201003} and \dref{2020630200} to get
 $\mathbb{P}^{-1}( {z}, { \psi })^{\top}\in D(\mathscr{A} )$.
   Consequently,  $D(\mathscr{A}_{\mathbb{P}})\subset \mathbb{P}D(\mathscr{A} )$ due to the arbitrariness of $(z,\psi)^{\top}\in D(\mathscr{A}_{\mathbb{P}} )$.

 On the other hand, for any $(z, \psi )^{\top}\in D(\mathscr{A}  )$,
\dref{2020613844D} and \dref{20206301054} imply    that $z\in Z_{F_1}$. Furthermore, it  follows from \dref{2020711041}  that
\begin{equation} \label{2020712118}
 \begin{array}{ll}
  \disp
    \tilde{G}_{\mu }  (\psi+P_{\mu}z)&\disp= \tilde{G}_{\mu }  \psi +P_{\mu}\tilde{A}z +
  B_{\mu}C_{\Lambda}z\crr
  &\disp= \tilde{G}_{\mu }  \psi -P_{\mu}F_1C_{\mu\Lambda}\psi+
  B_{\mu}C_{\Lambda}z+P_{\mu}(\tilde{A} z+F_1C_{\mu\Lambda} \psi),
  \end{array}
 \end{equation}
 which, together with  \dref{2020613844D}, \dref{20206141839} and \dref{202071808}, leads to
 \begin{equation} \label{20208201058}
 \tilde{G}_{\mu }  (\psi+P_{\mu}z) \in  L^2([0,\mu ];Y).
 \end{equation}
 This implies that $(\psi+P_{\mu}z)\in D(C_{\mu\Lambda})$.
 Since
$ \psi\in D(C_{\mu\Lambda})$, we have   $P_{\mu}z\in D(C_{\mu\Lambda})$. As a result,
    it follows from \dref{2020613844D} that
 \begin{equation} \label{20206271450Ad630}
(\tilde{A}-F_1C_{\mu\Lambda}P_{\mu})z+F_1C_{\mu\Lambda}(\psi+P_{\mu}z)=
\tilde{A} z+F_1C_{\mu\Lambda} \psi \in Z.
\end{equation}
Combining \dref{20206301117}, \dref{20208201058} and \dref{20206271450Ad630}, we arrive at
   $ \mathbb{P}( {z}, { \psi }) ^{\top}\in D(\mathscr{A}_{ \mathbb{P}} )$ and hence $  \mathbb{P} D(\mathscr{A} )\subset  D(\mathscr{A}_{ \mathbb{P}} )$.

To sum up, we   thus obtain
$  \mathbb{P} D(\mathscr{A} ) =  D(\mathscr{A}_{ \mathbb{P}} )$.
For any
$(z, \psi )^{\top}\in D(\mathscr{A}_{\mathbb{P}})$, it follows from \dref{20206301117} and  \dref{20206301054} that
$ z \in Z_{F_1}$.
By virtue of \dref{2020711041}, a simple computation shows that
 $\mathbb{P}\mathscr{A} \mathbb{P} ^{-1}(z, \psi )^{\top}=\mathscr{A}_{\mathbb{P}}(z, \psi )^{\top}$ for   any
$(z, \psi )^{\top}\in D(\mathscr{A}_{\mathbb{P}})$. Hence, the similarity \dref{20206301557} holds.
The proof is complete.
\end{proof}

By Lemma \ref{Lm2020615849}, the observer \dref{20206142114} is  convergent provided
 $\mathscr{A}_{\mathbb{P}}$ is stable.
Owing to the upper-block-triangular structure of $\mathscr{A}_{\mathbb{P}}$ and since $\tilde{G}_{\mu}$ is exponentially stable already, we only need to choose
$F_1$ such that $\tilde{A}-F_1C_{\mu\Lambda}P_{\mu}$ stable.  By   Lemma \ref{lm2020613824}
and \dref{202071900}, we can choose
$F_1\in \mathcal{L}(Y,[D(A^*)]')$
by the following scheme:
\begin{equation} \label{20206301524}
\left\{\begin{array}{ll}
{\rm (i)}   & {\rm   choose } \  F \in \mathcal{L}(Y,[D(A^*)]') \  {\rm   to\  detects }
 \ (A,C)\ {\rm  exponentially};\crr
{\rm (ii)}  &  {\rm   let }\  F_1= e^{ \tilde{A}\mu }F .
\end{array}\right.
\end{equation}
Under  \dref{20206301524}, the observer \dref{20206142114}
is found to be
\begin{equation} \label{202072805}
 \left\{\begin{array}{l}
\disp \dot{\hat{z}} (t)=\tilde{A} \hat{z} (t)
- e^{ \tilde{A}\mu }F[y(t)-C_{\mu\Lambda}\hat{\psi}(\cdot,t)]+B u(t),\crr
\hat{\psi}_t(\cdot,t)= \tilde{G}_{\mu} \hat{\psi}(\cdot,t)+B_{\mu}C_{ \Lambda}\hat{z}(t)
+P_{\mu}  e^{ \tilde{A}\mu }F
[y(t)-C_{\mu\Lambda}\hat{\psi}(\cdot,t)],
\end{array}\right.
\end{equation}
 or equivalently,
 \begin{equation} \label{20206142120}
 \left\{\begin{array}{l}
\disp \dot{\hat{z}} (t)=A \hat{z} (t)-e^{ \tilde{A}\mu }F[y(t)- \hat{\psi}(\mu,t)]+B u(t),\crr
\hat{\psi}_t(x,t)+ \hat{\psi}_x(x,t)=P_{\mu}e^{ \tilde{A}\mu }F
[y(t)- \hat{\psi}(\mu,t)],\crr
\disp \hat{\psi}(0,t)=C_{ \Lambda}\hat{z}(t),
\end{array}\right.
\end{equation}
where $P_{\mu}$ is given by \dref{202071733} and  $G_{\mu }$,  $B_{\mu  }$ and $C_{\mu  }$  are  defined by \dref{20205292012},
\dref{20205302223}  and  \dref{20205302118} with $\alpha=\mu$, respectively.

\begin{theorem}\label{Th2020731704}
Let  $(A ,B ,C )$ be   a regular linear system with the  state space $Z$, input space $U$ and output space $Y$. Suppose that $\mu>0$,
  $F \in \mathcal{L}(Y,[D(A^*)]')$    detects system  $(A,C)$ exponentially  and
  \begin{equation} \label{202082011061}
  (s-\tilde{A})^{-1}e^{\tilde{A}\mu}F\subset D(C_{\Lambda}) \ \ \mbox{for some }\  s\in \rho(A).
  \end{equation}
  Then,
      the observer \dref{20206142120}  of system \dref{wxh2020321104}  is well-posed:
  For any $(\hat{z} (0),\hat{\psi} (\cdot,0))^{\top}\in
  \mathcal{Z}_{\mu}(Y)  $ and $u\in L^2_{\rm loc}([0,\infty);U)$,
 the observer \dref{20206142120}   admits a unique solution $(\hat{z} ,\hat{\psi} )^{\top}\in C([0,\infty); \mathcal{Z}_{\mu}(Y)  )$  such that
\begin{equation} \label{2020615851}
 e^{\omega t} \|(z(t)-\hat{z} (t), \psi(\cdot,t)-\hat{\psi}(\cdot ,t))^{\top}\|_{ \mathcal{Z}_{\mu}(Y)  }\to  0
  \ \ \mbox{as}\ \ t\to\infty,
\end{equation}
   where       $\omega $  is a positive constant that is independent of $t$.
\end{theorem}
 \begin{proof}
 Let $F_1= e^{ \tilde{A}\mu }F$. Then \dref{202082011061} implies that
 $ Z_{F_1}\subset D(C_{\Lambda})$, where $Z_{F_1}$  is defined by  \dref{20206301054}.
 By Lemma \ref{lm2020516938}, the operator $P_{\mu}$ in \dref{20206142120} is well defined.

 Since $F  $    detects system  $(A,C)$ exponentially,
 it follows from
Lemma \ref{lm2020613824} that $ e^{ \tilde{A}\mu }F$ detects $(A,C_{\Lambda}e^{-A\mu })$ exponentially.
As a result,  the operator $\tilde{A}-F_1C_{\mu\Lambda}P_{\mu}$
generates an exponentially stable $C_0$-semigroup $e^{(\tilde{A}-F_1C_{\mu\Lambda}P_{\mu})t}$ on $Z$ and moreover, $F_1$ is admissible for  $e^{(\tilde{A}-F_1C_{\mu\Lambda}P_{\mu})t}$.
Since $\tilde{G}_{\mu}$ is exponentially stable already and $C_{\mu}$ is
admissible for  $e^{\tilde{G}_{\mu}t}$, it follows from Lemma \ref{Lm20207211026} in Appendix  that the operator
$\mathscr{A}_{\mathbb{P}}$ defined by \dref{20206301117} generates
an exponentially stable $C_0$-semigroup $e^{ \mathscr{A}_{\mathbb{P}} t}$ on $ \mathcal{Z}_{\mu}(Y) $.
 By Lemma \ref{Lm2020615849}, $\mathscr{A} $ and  $\mathscr{A}_{\mathbb{P}}$ are similar each other.
Therefore, the operator
$\mathscr{A} $ defined by \dref{2020613844} generates
an exponentially stable $C_0$-semigroup $e^{ \mathscr{A}  t}$ on $ \mathcal{Z}_{\mu}(Y) $.
As a result, the following system
\begin{equation} \label{2020731622}
 \left\{\begin{array}{l}
\disp \dot{\tilde{z}} (t) =A\tilde{z} (t)+e^{ \tilde{A}\mu }FC_{\mu\Lambda}\tilde{\psi}(\cdot,t),\crr
\disp  \dot{\tilde{\psi}}_t(\cdot,t) = ( G_{\mu}-P_{\mu} e^{ \tilde{A}\mu }F   C_{\mu\Lambda}) \tilde{\psi}(\cdot,t)+B_{\mu}C_{ \Lambda}\tilde{z} (t)
\end{array}\right.
\end{equation}
   with  initial state
\begin{equation} \label{2020731625}
 \left.\begin{array}{l}
\disp  \tilde{z}(0)=z(0)-\hat{z}(0), \ \ \
\tilde{\psi}  (\cdot,0)= \psi  (\cdot,0)-\hat{\psi}  (\cdot,0)
\end{array}\right.
\end{equation}
admits a unique solution $(\tilde{z},\tilde{\psi})^{\top}\in C([0,\infty);  \mathcal{Z}_{\mu}(Y)  )$ that decays exponentially to zero in $ \mathcal{Z}_{\mu}(Y) $  as $t\to\infty$.

 By Theorem \ref{th2020613745},
   system \dref{wxh2020321120}  admits a
 unique solution $(z ,\psi )^{\top}\in C([0,\infty);\mathcal{Z}_{\mu}(Y))$  for any $(z(0),\psi(\cdot,0))^{\top}\in \mathcal{Z}_{\mu}(Y)$ and $u\in L^2_{\rm loc}([0,\infty);U)$.
Let   $\hat{z}(t)=z(t)-\tilde{z}(t)$ and $\hat{\psi}(t)=\psi(t)-\tilde{\psi}(t)$. Then, such a defined $(\hat{z}(t),\hat{\psi}(t))^{\top}$ is a solution of  observer  \dref{20206142120}.
 Owing to the linearity of     \dref{20206142120}, the solution is unique.
Since system \dref{2020731622} happens to be the error system between   system
\dref{wxh2020321104} and its observer  \dref{20206142120},   the convergence \dref{2020615851}  holds.
The proof is complete.
\end{proof}
The assumption \dref{202082011061}  seems  a bit awkward.
When
$F$ or $C$ is a bounded operator, this assumption  can be   deduced from   that
 $F \in \mathcal{L}(Y,[D(A^*)]')$    detects system  $(A,C)$ exponentially.
When $F$ and $C$ are unbounded, we have the following Corollary:
 \begin{corollary}\label{Co2020615849}
Let  $(A ,B ,C )$ be   a regular linear system with the  state space $Z$, input space $U$ and output space $Y$. Suppose that
  $F \in \mathcal{L}(Y,[D(A^*)]')$    detects system  $(A,C)$ exponentially.
  Then,
      the observer \dref{20206142120}  of system \dref{wxh2020321104}   is well-posed for almost every $\mu>0$. That is,
  for any $(\hat{z} (0),\hat{\psi} (\cdot,0))^{\top}\in
  \mathcal{Z}_{\mu}(Y)  $ and $u\in L^2_{\rm loc}([0,\infty);U)$,
 the observer \dref{20206142120}   admits a unique solution $(\hat{z} ,\hat{\psi} )^{\top}\in C([0,\infty); \mathcal{Z}_{\mu}(Y)  )$  such that
\dref{2020615851}
 holds for some       positive constant $\omega $.
\end{corollary}
\begin{proof}
  Since  $(A,F,C)$
is a  regular linear system, $C$ is admissible for $e^{ At}$.
This implies that $e^{ A\mu}Z \subset  D(C_{\Lambda}) $ for almost every $\mu>0$.
Consequently,
   \begin{equation} \label{202072826}
   (s-\tilde{A})^{-1}e^{\tilde{A}\mu}FY= e^{  {A}\mu}(s-\tilde{A})^{-1}FY
  \subset e^{ {A}\mu}Z
  \subset
  D(C_{\Lambda}),\ \ \forall\ s\in\rho(A),
  \end{equation}
 which, together with Theorem \ref{Th2020731704}, completes the proof.
\end{proof}



\begin{remark}\label{Re2020715178}
When $A$ is a  matrix,  it follows from \dref{20206141839} and \dref{202071808} that
 the  observer \dref{20206142120} takes form
 \begin{equation} \label{20207161051}
 \left\{\begin{array}{l}
\disp \dot{\hat{z}} (t)=A \hat{z} (t)-e^{A \mu }F[y(t)- \hat{\psi}(\mu,t)]+B u(t),\crr
\hat{\psi}_t(x,t)+ \hat{\psi}_x(x,t)= -C  e^{ A(\mu-x) }F
[y(t)- \hat{\psi}(\mu,t)],\crr
\disp \hat{\psi}(0,t)=C \hat{z}(t),
\end{array}\right.
\end{equation}
which is the same as the observer in \cite{Krstic2008delay}.
In contrast to the PDE backstepping method used in \cite{Krstic2008delay},
we never need the target system.
 Moreover,
  the Lyapunov function  has not used  in the
    proof of observer convergence. Once again as  Remark
\ref{Re2020716940},  this  gives a way to avoid the difficulties in construction of the Lyapunov
functional  for PDEs with  delay.
 \end{remark}

\section{Application to 1-D wave equation}\label{Application}
To show the effectiveness of the developed approach, we apply the abstract results to the benchmark wave equation:
   \begin{equation} \label{2020515700}
\left\{\begin{array}{l}
     \disp  z_{tt}(\sigma,t)=z_{\sigma\sigma}(\sigma,t),\ \ \sigma\in(0,1),\crr
  \disp     z(0,t)=0,\ z_\sigma(1,t)=u(t-\tau),
        \end{array}\right.
 \end{equation}
 where $u(t)$  is the control input which suffers from
 a time-delay $\tau >0$.
 The input space is  $  \R$ and
 the state space  is
 $Z=\{(f,g)\in H^1(0,1)\times L^2(0,1)\ |\  f(0)=0\}$
  with the inner product
  \begin{equation} \label{20205131429}
 \langle (f_1,g_1),(f_2,g_2)\rangle_{Z}=\int_{0}^1 f_1'(x)  {f_2'(x)}
 +g_1(x) {g_2(x)}dx,\ \
 \forall\ (f_i,g_i)\in Z,\ i=1,2.
 \end{equation}
Define the operator    $A : D(A )\subset Z\to Z$ by
 \begin{equation} \label{2020514722}
\left\{\begin{array}{l}
     \disp    A (f,g)^{\top}=(g,f'')^{\top},\ \ \forall\ (f,g)^{\top}\in D(A ),\crr
  \disp      D(A )=\{ (f,g)\in H^2(0,1)\times H^1(0,1)  \ |\  f(0)=g(0)=0,f'(1)=0\}.
       \end{array}\right.
 \end{equation}
In view of \dref{20205292010},
  system \dref{2020515700} can be written as the form
 \begin{equation} \label{20207141102}
\left\{\begin{array}{l}
\disp \frac{d}{dt}(z(\cdot,t),z_t(\cdot,t))^{\top}=A(z(\cdot,t),z_t(\cdot,t))^{\top}+B\phi(\tau,t),\crr
\disp    \phi_t(x,t)+ \phi_x(x,t)=0,\ \  x\in(0,\tau), \crr
\disp \phi(0,t)=u(t),
\end{array}\right.
\end{equation}
  where
  the control operator  $B =(0,\delta(\cdot-1))^{\top}$ and $\delta(\cdot)$ is the Dirac distribution.
Let
\begin{equation} \label{2020716720}
K (f,g)^{\top}=-k_1g(1)\ \ \mbox{ for any}\ \  (f,g)^{\top}\in D(A ),\ \ k_1>0.
\end{equation}
Then,   $K=-k_1B ^*$ and   it is well known that
 $K $ stabilizes system  $(A,B)$ exponentially. In view of \dref{2020611024}, we obtain the feedback
  \begin{equation} \label{202075909}
u(t)=K_{\Lambda}  \int_0^{\tau} e^{\tilde{A}  x} B \phi (x,t)dx+K_{\Lambda} e^{  {A} \tau}({z}(\cdot,t),z_t(\cdot,t))^{\top}.
\end{equation}
We next seek the analytic form of the feedback.
Let $z_{\delta}=(\sigma,0)^{\top}$ with $\sigma\in [0,1]$.
  A simple computation shows that
    $\tilde{A}z_{\delta}=B$ and \begin{equation}\label{2020751456}
e^{Ax}z_{\delta}
=\left(\sum\limits_{n=0}^{\infty}(-1)^{n}\frac{2}{\omega_n^{2}}\cos\omega_nx
\sin\omega_n\sigma,  \sum\limits_{n=0}^{\infty}(-1)^{n+1}\frac{2}{\omega_n}
\sin\omega_nx\sin \omega_n\sigma \right)^{\top},
\end{equation}
where
\begin{equation} \label{2020781634}
\omega_n=\frac{(2n+1)\pi  }{2}, \ \ \sigma\in [0,1],\   x\in[0,\tau],\ \ n=0,1,2,\cdots.
\end{equation}
Moreover, it follows from   \dref{2020751456}   that
\begin{equation} \label{2020781646}
\int_0^{\tau} e^{Ax}z_{\delta} \phi (x,t)dx=
\left(
\sum\limits_{n=0}^{\infty}(-1)^{n}\frac{2\alpha_n(t)}{\omega_n^{2}}
\sin\omega_n\sigma , \sum\limits_{n=0}^{\infty}(-1)^{n+1}\frac{2\beta_n(t)}{\omega_n}
 \sin \omega_n\sigma
\right)^{\top} ,
\end{equation}
where
\begin{equation} \label{2020781649}
\alpha_n(t)=\int_0^{\tau}
 \cos\omega_nx \phi(x,t)dx\ \ \
\beta_n(t)=\int_0^{\tau}
 \sin\omega_nx \phi(x,t)dx.
\end{equation}
Since $B$ is admissible for $e^{At}$ and $\phi(\cdot,t)\in L^2(0,\tau)$, we have
\begin{equation} \label{202076949}
\begin{array}{ll}
\disp \int_0^{\tau} e^{ \tilde{A}x }B\phi (x,t)dx&\disp=\int_0^{\tau} e^{ \tilde{A}x }\tilde{A}\tilde{A}^{-1}B\phi (x,t)dx=  \int_0^{\tau}e^{ \tilde{A}x }\tilde{A} z_{\delta}\phi (x,t)dx\crr
&\disp =  \int_0^{\tau} \tilde{A} e^{ {A}x }z_{\delta}\phi (x,t)dx
=  \tilde{A}\int_0^{\tau} e^{ {A}x }z_{\delta}\phi (x,t)dx\in Z,
\end{array}
\end{equation}
which implies that
\begin{equation} \label{20207141500}
\int_0^{\tau} e^{Ax}z_{\delta} \phi (x,t)dx\in D(A).
\end{equation}
  Combining \dref{2020514722}, \dref{2020716720},   \dref{20207141500},
  \dref{202076949}, \dref{2020781646} and \dref{2020781649}, we arrive at
\begin{equation} \label{2020781713}
\begin{array}{ll}
\disp    K_{\Lambda}\int_0^{\tau} e^{\tilde{A}  x} B \phi (x,t)dx
&\disp =    K_{\Lambda} A \int_0^{\tau} e^{ {A}x }z_{\delta} \phi (x,t)dx
 = -2k_1\sum\limits_{n=0}^{\infty}  \alpha_n (t) .
 \end{array}
\end{equation}
By a straightforward  computation, we have
\begin{equation} \label{2020781719}
\begin{array}{ll}
\disp  K_{\Lambda} e^{  {A} \tau}({z}(\cdot,t),z_t(\cdot,t))^{\top}&\disp =
 -k_1
 \sum\limits_{n=0}^{\infty} (-1)^n\omega_n\left[ \zeta_n(t)\cos\omega_n\tau-\gamma_n(t) \sin\omega_n\tau\right],
\end{array}
\end{equation}
where
\begin{equation} \label{2020781725}
\gamma_n(t)=2\int_0^1z (\sigma,t)\sin\omega_n\sigma d\sigma ,\ \
 \zeta_n(t) =\frac2{\omega_n}\int_0^1z_{t}(\sigma,t)\sin\omega_n\sigma d\sigma ,\  n=0,1,2,\cdots.
\end{equation}

By \dref{202075909}, \dref{2020781713} and  \dref{2020781719}, we get the closed-loop system \begin{equation} \label{2020781730}
 \left\{\begin{array}{l}
 \disp  z_{tt}(\sigma,t)=z_{\sigma\sigma}(\sigma,t),\ \ \sigma\in(0,1),\crr
  \disp     z(0,t)=0,\ z_\sigma(1,t)=\phi( \tau,t),  \crr
\disp    \phi_t(x,t)+ \phi_x(x,t)=0, \ \   x\in(0,\tau), \crr
\disp \phi(0,t)
\disp =-2k_1 \sum\limits_{n=0}^{\infty}  \alpha_n (t)
\disp  -k_1
 \sum\limits_{n=0}^{\infty} (-1)^n\omega_n\left[ \zeta_n(t)\cos\omega_n\tau-\gamma_n(t) \sin\omega_n\tau\right],
 \end{array}\right.
\end{equation}
where   $k_1>0$, $\alpha_n(t)$ is given by  \dref{2020781649}     and  $\zeta_n(t),\gamma_n(t)$
are given  by \dref{2020781725}. By Theorem \ref{Th2020611038}, the solution of   closed-loop system  \dref{2020781730} is well posed and decays to zero exponentially as $t\to\infty$.

\begin{remark}\label{Re20207111229}
The infinite  series  in the closed-loop system  \dref{2020781730}
   can also  be written  as a  dynamic form.  Actually,
   a simple computation shows that
 \begin{equation} \label{20207271004}
v_1(\cdot,\tau;t)=  \int_0^{\tau} e^{\tilde{A}  s} B \phi (s,t)ds,
\end{equation}
   where
 \begin{equation} \label{20207271005}
 \left\{\begin{array}{l}
  \disp v_{1xx}(\sigma ,x;t)=v_{1\sigma \sigma }(\sigma ,x;t),  \ \ \sigma \in (0,1), \ 0<x\leq\tau,  \crr
 \disp v_{1}(0,x;t)=0,\ v_{1\sigma }(1,x;t)=\phi(\tau-x,t),\crr
   (v_1(\sigma ,0;t), v_{1x}(\sigma ,0;t))\equiv(0,0),\ \ \sigma \in [0,1].
\end{array}\right.
\end{equation}
The    notation  $v_1(\cdot,\cdot;t)$ means that the function $v_1$  depends on the time $t$. If we let
\begin{equation} \label{2020720919}
v_2(\cdot,x;t)=    e^{\tilde{A}  x}  ({z}(\cdot,t),z_t(\cdot,t))^{\top},
\end{equation}
  then it is  governed by
\begin{equation} \label{2020720921}
 \left\{\begin{array}{l}
  \disp v_{2xx}(\sigma ,x;t)=v_{2\sigma \sigma }(\sigma ,x;t),  \ \ \sigma \in (0,1), \ x>0,  \crr
 \disp v_{2}(0,x;t)= v_{2\sigma }(1,x;t)=0,\crr
 (v_2(\sigma ,0;t), v_{2x}(\sigma ,0;t))=({z}(\sigma ,t),z_t(\sigma ,t)),\ \ \sigma \in [0,1].
\end{array}\right.
\end{equation}
 Combining \dref{20207271004}, \dref{20207271005}, \dref{2020720921}, \dref{202075909} and \dref{2020716720}, we obtain the following closed-loop system:
 \begin{equation} \label{202075909dynamics}
 \left\{\begin{array}{l}
   \disp  z_{tt}(\sigma,t)=z_{\sigma\sigma}(\sigma,t),\ \ \sigma\in(0,1), \crr
  \disp     z(0,t)=0,\ z_\sigma(1,t)=\phi( \tau,t),   \crr
\disp    \phi_t(x,t)+ \phi_x(x,t)=0, \ \ \ x\in(0,\tau), \crr
\disp \phi(0,t)=-k_1v_{1x}(1,\tau;t) -k_1v_{2x}(1,\tau;t),\ \ k_1>0,\crr
  v_{1}(\cdot,\cdot;t),v_2(\cdot,\cdot;t)\  \mbox{are given by}\  \dref{20207271005} \ \mbox{and}\  \dref{2020720921}.
\end{array}\right.
\end{equation}

\end{remark}

\vskip0.3cm
 Now we consider the output delay compensation for the wave equation in \dref{2020515700}.  Suppose that   we can measure
the average velocity  around $\sigma_0\in(0,1)$ and the output is
  \begin{equation} \label{2020772152}
  y(t)=\int_0^1m(\sigma)z_t(\sigma,t-\mu)d\sigma,\ \ \mu>0,
  \end{equation}
 where   $m\in L^2(0,1)$ is the shaping function
 around   the sensing point
$\sigma_0$. System \dref{2020515700} with output \dref{2020772152} can be  written as
 \begin{equation} \label{2020941115}
\left\{\begin{array}{l}
     \disp  z_{tt}(\sigma,t)=z_{\sigma\sigma}(\sigma,t),\ \ \sigma\in(0,1),\crr
  \disp     z(0,t)=0,\ z_\sigma(1,t)=u(t-\tau),\ \ \tau\geq0,\crr
\disp     \psi_t(x,t)+ \psi_x(x,t)=0\  \  \ \ x\in[0,\mu ], \ \ \mu>0,\crr
\disp  \psi(0,t)=\int_0^1m(s)z_t(s,t)ds,\crr
\disp y(t)= \psi(\mu ,t).
\end{array}\right.
\end{equation}
The observation
 operator $C$ is given by
\begin{equation} \label{2020772222}
 C: (f,g)^{\top}\to \int_0^1m(\sigma)g(\sigma) d\sigma,\ \ \forall\ \ (f,g)^{\top}\in Z.
  \end{equation}
It is evident that  $C$ is
bounded. We choose $m$ such that $(A,C)$ is exactly observable.
 If we let $F=-k_2C^*$, $k_2>0$, then $F\in\mathcal{L}(\R,Z)$ is given by $Fq=-k_2q(0,m(\cdot))^{\top}$
 for any $q\in \R$. Since  $F$ detects system $(A,C)$ exponentially \cite{LiuKS1997},
 by \dref{20206142120}, the observer of system \dref{2020941115} is
\begin{equation} \label{2020751656}
 \left\{\begin{array}{l}
\disp \frac{d}{dt}(\hat{z}(\cdot,t),\hat{z}_t(\cdot,t))^{\top}=A (\hat{z}(\cdot,t),\hat{z}_t(\cdot,t))^{\top}-e^{  {A}\mu }F[y(t)-\hat{\psi}(\mu,t)]+B u(t-\tau),\crr
\hat{\psi}_t(x,t)+ \hat{\psi}_x(x,t)=P_{\mu}e^{  {A}\mu }F
[y(t)-\hat{\psi}(\mu,t)],\crr
\disp \hat{\psi}(0,t)=\int_0^1m(\sigma)\hat{z}_t(\sigma,t)d\sigma.
\end{array}\right.
\end{equation}
Since
  \begin{equation} \label{20205141058}
e^{ \tilde{A}x }F=e^{ {A}x }F=-k_2\left(  \sum_{n=0}^{\infty}
f_n\sin\omega_n x \sin \omega_n\sigma,
\sum_{n=0}^{\infty}
f_n \omega_n  \cos\omega_n x \sin \omega_n\sigma
\right)^{\top},
\end{equation}
 where $0\leq x\leq\mu $, $0\leq \sigma\leq1$, $\omega_n$ is given by \dref{2020781634} and
  \begin{equation} \label{2020781420}
f_n= \frac{2}{\omega_n}\int_{0}^{1}m(\sigma)
\sin \omega_n\sigma  d\sigma,\ n=0,1,2,\cdots,
\end{equation}
it follows from  \dref{202071808} and \dref{20206141839} that
   \begin{equation} \label{2020751019}
   \begin{array}{rl}
\disp P_{\mu }e^{ \tilde{A}\mu }F&\disp =  - Ce^{A(\mu-x)}F=
k_2\int_{0}^{1}
\left(\sum_{n=0}^{\infty}f_n\omega_n\cos  \omega_n(\mu-x) \sin \omega_n\sigma \right)m(\sigma)d\sigma\crr
&\disp =k_2 \sum_{n=0}^{\infty}\frac{f_n^2\omega_n^2}{2} \cos  \omega_n(\mu-x).
\end{array}
\end{equation}
Since
\begin{equation} \label{2020781517}
  \sum_{n=0}^{\infty}\frac{f_n^2\omega_n^2}{2}  =
  2\sum_{n=0}^{\infty} \left|\int_0^1m(\sigma)\sin\omega_n\sigma d\sigma\right|^2<+\infty,
\end{equation}
the series in \dref{2020751019} is convergent.
Combining  \dref{20205141058} and \dref{2020751019}, the observer \dref{2020751656}  becomes
\begin{equation} \label{2020781535}
 \left\{\begin{array}{l}
\disp  \hat{z}_{1t}(\sigma,t)=\hat{z}_2(\sigma,t)+k_2
 \left(  \sum_{n=0}^{\infty}
f_n\sin\omega_n \mu \sin \omega_n\sigma\right)[y(t)-\hat{\psi}(\mu,t)],\crr
\disp  \hat{z}_{2t}(\sigma,t)= \hat{z}_{1\sigma\sigma}(\sigma,t)+k_2
\left(\sum_{n=0}^{\infty}
f_n \omega_n  \cos\omega_n \mu \sin \omega_n\sigma\right)[y(t)-\hat{\psi}(\mu,t)]+  u(t-\tau),\crr
 \disp\hat{\psi}_t(x,t)+  \hat{\psi}_x(x,t)=k_2\left[\sum_{n=0}^{\infty}\frac{f_n^2\omega_n^2}{2} \cos  \omega_n(\mu-x)\right]
[y(t)-\hat{\psi}(\mu,t)],\crr
\disp \hat{\psi}(0,t)=\int_0^1m(\sigma)\hat{z}_t(\sigma,t)d\sigma,
\end{array}\right.
\end{equation}
where $k_2>0$, $\hat{z}_1(\sigma,t)=\hat{z}(\sigma,t)$ and
$\hat{z}_2(\sigma,t)=\hat{z}_t(\sigma,t)$ for any  $\sigma\in[0,1]$ and $t\geq0$.
By Theorem \ref{Th2020731704},  $(\hat{z}_1(\cdot,t),\hat{z}_2(\cdot,t)) $ converges to
$(z(\cdot,t),z_t(\cdot,t)) $ exponentially in $Z$ as $t\to\infty$.

\begin{remark}\label{Re2020720936}
By  our abstract theory, the proposed approach is still working for the
  delayed  boundary output $y(t)=z(1,t-\mu)$. In this case, we can choose
 $F =-k_2(0,\delta(\cdot-1))^{\top}$ with $k_2>0$. However,
since $F$ is unbounded now,  it is not easy to obtain the analytic forms
of  the   gain operators  $e^{ \tilde{A}\mu }F$ and
$P_{\mu}e^{ \tilde{A}\mu }F$
 in the state space. Therefore, a further effort is   still needed
  for the observer design of infinite-dimensional systems with   unbounded   delayed output.

 \end{remark}



\section{Appendix}
\begin{lemma}\label{lm2020632124}
Suppose that $(A,B,C)$ is a linear system with the
state space $Z$, input space $U$ and the output space $Y$. For any $0\neq q\in\R$,
 the following assertions are true:

(i).  If  $C\in \mathcal{L}(D(A),Y)$ is admissible for
$e^{At}$,  then  $Ce^{Aq} \in \mathcal{L}(D(A),Y) $ is admissible for
$e^{At}$ as well;

(ii). If $B\in \mathcal{L}(U, [D(A^*)]')$  is admissible for
$e^{At}$, then    $ e^{\tilde{A}q}B \in \mathcal{L}(U, [D(A^*)]') $ is admissible for
$e^{At}$ as well.

\end{lemma}
\begin{proof}
  For any  $z_0\in D(A)$ and
  $\tau>0$,  we have
\begin{equation} \label{2020632130}
      \int_{0}^{\tau}\|C  e^{Aq }e^{At}z_0\|_Y^2dt
   =\int_{q}^{q+\tau}\|C  e^{As }z_0\|_Y^2ds
   \leq
   \int_{0}^{q+\tau }\|C  e^{A  s  }z_0\|_Y^2ds.
  \end{equation}
Since $C\in \mathcal{L}(D(A),Y)$ is admissible for
$e^{At}$,   there exists an  $M>0$ such that
 \begin{equation} \label{2020632141}
       \int_{0}^{\tau+q}\|C  e^{As }z_0\|_Y^2ds\leq M\|z_0\|_Z^2,
  \end{equation}
which, together with \dref{2020632130},  shows  that
$Ce^{Aq} \in \mathcal{L}(D(A),Y) $ is admissible for
$e^{At}$.

  Since $B $  is admissible for
$e^{At}$, $B^*  $  is admissible for
$e^{A^*t}$. By  (i) just proved, $B^*  e^{A^*q }$ is admissible for
$e^{A^*t}$. Hence,  $  e^{\tilde{A} q }B$ is admissible for
$e^{A t}$. This completes the proof.
\end{proof}

 \begin{lemma}\label{Lm201912031723}
 Let  $(A,C) $ be a linear system with the state space $Z$ and output space $U$.
 Suppose that $C$  is admissible for $e^{At}$.
 Let $G_{\alpha}$  and $B_{\alpha}$   be   given by
 \dref{20205292012} and \dref{20205302223},   respectively.
  Define  $\mathcal{Z}_{\tau  }(U) =Z\times L^2([0,\tau];U)$ and
  \begin{equation} \label{2020428921}
 \A =\begin{pmatrix}
\tilde{A} &0\\B_{\alpha} C_{ \Lambda}&\tilde{G}_{\alpha}
\end{pmatrix},
\left.\begin{array}{l}
 \disp D(\A )=\left\{\begin{pmatrix}
z\\g\end{pmatrix} \in \mathcal{Z}_{\tau  }(U)  \ \Big{|} \ \begin{array}{l}
 \tilde{A}z \in Z \\
  \tilde{G}_{\alpha} g+B_{\alpha}C_{\Lambda}z\in L^2([0,\alpha];U)
  \end{array}  \right\}.
 \end{array}\right.
 \end{equation}
  Then, the operator $\A $   generates a
 $C_0$-semigroup $e^{\A  t}$ on $\mathcal{Z}_{\tau  }(U)$.
 Moreover, if we suppose  further that  $e^{A t}$ is
   exponentially stable in  $Z$, then  $e^{\A  t}$ is  exponentially stable in
   $\mathcal{Z}_{\tau  }(U)$.
\end{lemma}
\begin{proof}

The operator $\A $ is associated with the following system:
\begin{equation}\label{20191311712FHAD}
\left\{\begin{array}{ll}
  \disp  \dot{z} (t)=A z(t),\crr
   \disp   {\phi}_t (\cdot,t)+{\phi}_x(\cdot,t)=0\ \ \mbox{in}\ \ U,\ \ \
   \phi(0,t)=C_{\Lambda}z(t).
\end{array}\right.
\end{equation}
Since $C$  is admissible for  the semigroup $e^{At}$, for any  $(z(0),\phi(\cdot,0)) ^{\top}\in D(\A)$,
we have $z(0)\in D(A)$, $z\in C^1([0,\infty); Z)  $,  $ C_{\Lambda}z \in {H}^1_{\rm loc}([0,\infty); U) $ and
\begin{equation}\label{20208141843}
\phi(x,t)=\left\{\begin{array}{ll}
   C_{\Lambda}z(t-x),& t-x\geq0,\\
  \phi(x-t,0),&t-x<0,
\end{array}\right. \ \ x\in[0,\alpha] .
\end{equation}
Therefore, system \dref{20191311712FHAD} admits  a unique continuously differentiable  solution  $(z,\phi)^\top\in C^1([0,\infty);\mathcal{Z}_{\tau}(U))$ for    any   $(z(0),\phi(\cdot,0)) ^{\top}\in D(\A)$. By \cite[Theorem 1.3, p.102]{Pazy1983Book}, the operator $\A $   generates a  $C_0$-semigroup $e^{\A  t}$ on   $ \mathcal{Z}_{\tau  }(U)$.

Finally, we show  the exponential stability.
 Suppose that $(z,\phi)^{\top}\in C([0,\infty);\mathcal{Z}_{\tau  }(U))$ is a solution of system
  \dref{20191311712FHAD}.
  Since
 $e^{A t}$  is  exponentially stable on $Z$ and  $\phi$-subsystem is independent of the $z$-subsystem,
  there exist two positive constants
$\omega_A$ and $L_A$ such that
  \begin{equation}\label{201994835FHAD}
\|z(t)\|_{Z}\leq L_Ae^{-\omega_At}\|z(0)\|_{Z},\ \ \forall\ t\geq0.
\end{equation}
Moreover, it follows from  \cite[Proposition 4.3.6, p.124]{TucsnakWeiss2009book} that
     \begin{equation}\label{201994731FHAD}
 v_{\omega}\in L^2([0, \infty );U),\ \ v_{\omega}(t)= e^{\omega t}C_{\Lambda}z(t),\ \ 0<\omega<\omega_A  ,
\end{equation}
 which, together with \dref{20208141843}, implies that $\phi(\cdot,t)$ decays to zero exponentially as $t\to\infty$. So $(z,\phi)$ decays to zero exponentially in $\mathcal{Z}_{\tau  }(U)$.
 The proof is complete.
\end{proof}

\begin{lemma}\label{Lm20207211026}
 Let  $(A,B) $ be a linear system with the state space $Z$ and input space $U$.
 Suppose that $B$  is admissible for $e^{At}$.
 Let $G_{\alpha}$   and $C_{\alpha}$ be   given by
  \dref{20205292012}  and \dref{20205302118},   respectively.
  Define  $\mathcal{Z}_{\tau  }(U) =Z\times L^2([0,\tau];U)$ and
  \begin{equation} \label{20207211026}
 \A =\begin{pmatrix}
\tilde{A} &BC_{\alpha\Lambda}\\0&\tilde{G}_{\alpha}
\end{pmatrix},
\left.\begin{array}{l}
 \disp D(\A )=\left\{\begin{pmatrix}
z\\g\end{pmatrix} \in \mathcal{Z}_{\tau  }(U)  \ \Big{|} \ \begin{array}{l}
 \tilde{A}z +BC_{\alpha\Lambda}g\in Z \\
  \tilde{G}_{\alpha} g \in L^2([0,\alpha];U)
  \end{array}  \right\}.
 \end{array}\right.
 \end{equation}
  Then,  $\A $   generates a
 $C_0$-semigroup $e^{\A  t}$ on $\mathcal{Z}_{\tau  }(U)$.
 Moreover, if we suppose  further that  $e^{A t}$ is
   exponentially stable in $Z$, then  $e^{\A  t}$ is  exponentially stable in
   $\mathcal{Z}_{\tau  }(U)$.
\end{lemma}
\begin{proof}
 Almost  the same as Lemma \ref{Lm201912031723},
 we can prove that the operator $\A $   generates a
 $C_0$-semigroup $e^{\A  t}$ on $\mathcal{Z}_{\tau  }(U)$. It suffices to prove the exponential
 stability.
 Consider the classical solution of the following system:
\begin{equation}\label{20207211035}
\left\{\begin{array}{ll}
  \disp  \dot{z} (t)=A z(t)+B\phi(\alpha,t),\crr
   \disp   {\phi}_t (\cdot,t)+{\phi}_x(\cdot,t)=0\ \ \mbox{in}\ \ U,\ \ \
   \phi(0,t)=0.
\end{array}\right.
\end{equation}
Since $e^{\tilde{G}_{\alpha}t}$ vanishes  after time $\alpha$ and $e^{A t}$ is
   exponentially stable in $Z$,  the solution  $(z,\phi)$ decays to zero exponentially in $\mathcal{Z}_{\tau  }(U)$ as $t\to\infty$.
 The proof is complete.
  \end{proof}
\end{document}